\documentclass[sn-mathphys-num]{sn-jnl}

\usepackage{graphicx}%
\usepackage{multirow}%
\usepackage{amsmath,amssymb,amsfonts}%
\usepackage{amsthm}%
\usepackage{mathrsfs}%
\usepackage[title]{appendix}%
\usepackage{xcolor}%
\usepackage{textcomp}%
\usepackage{manyfoot}%
\usepackage{booktabs}%
\usepackage{algorithm}%
\usepackage{algorithmicx}%
\usepackage{algpseudocode}%
\usepackage{listings}%

\usepackage{siunitx}
\usepackage{tikz}
\usetikzlibrary{decorations.text,calc,arrows.meta}
\usepackage{forest, tikz-qtree,pseudocode,ulem, framed,fancybox}
\newsavebox{\tempbox}
\newcommand{\cbox}[2]{%
  \fcolorbox{black}{#1}{\texttt{#2\strut}}\kern-\fboxrule}
\usepackage{times}
\usepackage{palatino}
\usepackage{epsfig}
\usepackage{xcolor}
\DeclareMathOperator{\Var}{Var}

\newtheorem{claim}{Claim}
\newtheorem{open}{Open problem}
\newtheorem{corollary}{Corollary}
\newtheorem{lemma}{Lemma}

\newtheorem{theorem}{Theorem}
\newtheorem{proposition}[theorem]{Proposition}%
\newtheorem{example}{Example}%
\newtheorem{remark}{Remark}%

\newtheorem{definition}{Definition}%

\begin{document}

\title[Interacting Particle Systems on hypergraphs]{Interactive Particle Systems on Hypergraphs, Drift Analysis and the $\mathit{WalkSAT}$ Algorithm}


\author*[1]{\fnm{Gabriel} \sur{Istrate}}\email{gabriel.istrate@g.unibuc.ro}

\author[2]{\fnm{Cosmin} \sur{Bonchi\c{s}}}\email{cosmin.bonchis@e-uvt.ro}

\author[2]{\fnm{Mircea} \sur{Marin}}\email{mircea.marin@e-uvt.ro}


\affil*[1]{\orgdiv{Faculty of Mathematics and Computer Science}, \orgname{University of Bucharest}, \orgaddress{\street{Str. Iuliu Maniu 15G}, \city{Bucharest}, \postcode{061075}, \state{Sector 6}, \country{Romania}}}

\affil[2]{\orgdiv{Department of Computer Science}, \orgname{West University of Timi\c{s}oara}, \orgaddress{\street{Bd. V. P\^{a}rvan 4}, \city{Timi\c{s}oara}, \postcode{300223}, \country{Romania}}}


\abstract{We rigorously analyze the expected running time of $\mathit{WalkSAT}$, a well-known local search procedure for satisfiability solving, on individual satisfiable instances of the $k$-XOR SAT problem. We obtain estimates of this expected running time by reducing the problem to a setting amenable to classical techniques from drift analysis. 

A crucial ingredient of our reduction is the definition of \textit{hypergraph versions of interacting particle systems}, notably of coalescing and annihilating random walks, as well as the voter model. These extensions are of independent interest. Their use allows to show that the expected running time of $\mathit{WalkSAT}$ depends on two structural parameters (called \textit{odd Cheeger drift} and \textit{Cheeger time}) of the (dual of the) formula hypergraph.}




\maketitle

\section{Introduction} 

Interacting particle systems are discrete dynamical systems, usually defined on lattices, studied intensely in Mathematical Physics  \cite{liggett-ips}. They can be investigated on finite graphs as well \cite{donnelly-welsh-interacting}, \cite{donnelly-welsh-coloring}, \cite{aldous1999deterministic} as finite Markov chains; some of them correspond via \textit{duality} to certain types of random walks \cite{liggett-ips}. It is, therefore, not that surprising that the analysis of such particle systems can  sometimes be used to upper bound the mixing time of several Markov chains, e.g. (hyper)graph coloring procedures \cite{donnelly-welsh-coloring,chung2012hypergraph}. 

In this paper we consider generalized versions of interacting particle systems in conjunction with a problem in the rigorous analysis of randomized search algorithms for combinatorial optimization. While progress has been made \cite{doerr2019theory}, tools and techniques for the analysis of such methods are still somewhat limited, and our theoretical knowledge still considerably lags that in the experimental investigation of heuristics. 

The problem we deal with is 
\textit{the analysis of the local search procedure WalkSAT} \cite{WalkSAT}  for a version of the satisfiability problem, the so-called $k$-XOR-SAT problem \cite{mzard2002alternative}, where $k\geq 2$.  We show that studying \textit{hypergraph analogues of coalescing/annihilating random walks and the voter model} allows the control of the expected convergence time of $\mathit{WalkSAT}$ on many \textbf{individual satisfiable instances} $\Phi$ of  $k$-XOR-SAT in terms of structural parameters of two associated hypergraphs, the formula hypergraph of $\Phi$ (for uniquely satisfiable instances), as well as a certain "dual" of this hypergraph (in general). The tool we employ is drift analysis \cite{he2001drift}. 

On a technical level, our main conceptual contribution is \textbf{the extension of classical versions of annihilating and coalescing random walks,  as well as the voter model to hypergraphs}. Besides their intended application to XOR-SAT, such  generalizations have obvious intrinsic interest, and add to the growing recent literature on extending interacting particle systems  to hypergraphs \cite{lu2012loose,cooper2012spectra,chung2012hypergraph,cooper2013cover,avin2014radio} and simplicial complexes \cite{steenbergen2012cheeger,parzanchevski2012isoperimetric}. The analysis we perform also has consequences for several other seemingly unrelated problems, such as   
 social balance \cite{redner-balance} and  lights-out games \cite{scherphuis-lightsout}. 
 Moreover, while inheriting some properties from the graph case, our hypergraph generalizations of interacting particle systems  have a number of features that make them interesting from a technical perspective; we refer the interested reader to the Conclusions section for a brief discussion. 

The plan of the paper is as follows: In Section~\ref{ps} we explain the models/problem we are interested in. In Section~\ref{prelim} we review some of the useful notions and prerequisites. In Section~\ref{motiv} we motivate some of the concepts relevant to our result through the analysis of a particular case. This example can easily be extended (Theorem~\ref{thm:3}) to many families of uniquely solvable instances of $k$-XOR-SAT. There are three cases, intuitively corresponding to positive, neutral and negative drift, respectively.  

The main contribution of the paper belongs to its second half: 
In section~\ref{duality} we reinterpret the dynamics underlying $\mathit{WalkSAT}$ by duality. 
Our main result (Theorem~\ref{cor}) applies extensions (we give as Theorem~\ref{vm}) of techniques developed in \cite{aldous-fill-book,aldous2013interacting} for the analysis of the voter model on finite graphs, bounding the expected convergence time of $\mathit{WalkSAT}$ in terms of two Cheeger-like constants of the dual of the formula hypergraph. The strategy of the proof is explained in Section~\ref{sec:6}, the necessary technical machinery is developed in Section~\ref{sec:7}, while the proof itself is given in Section~\ref{sec:8}.

\section{Problem Statement} 
\label{ps} 
We are concerned with a version of the satisfiability problem called \textit{$k$-XOR satisfiability} ($k$-XOR-SAT): 

\begin{definition} Given $k\geq 2$, 
 an instance of \textit{$k$-XOR-SAT} is a linear system $\Phi$ of $m$ equations $A\cdot \overrightarrow{x}=\overrightarrow{b}$ 
 over boolean variables, where
 $A$ is an $m\times n$ matrix, $m,n\geq 1$,  
$\overrightarrow{x}=(x_{1},x_{2},\ldots, x_{n})^{T}$ is an $n\times 1$ vector,  
$\overrightarrow{b}=(b_{1},b_{2},\ldots, b_{m})^{T}$ is an $m\times 1$ vector, and each equation has exactly $k$ variables. We  assume that such a system does \textit{not} contain multiple copies of the same equation. 

$\Phi$ is \textit{satisfied} by a vector $Y=(y_1,\ldots, y_n)\in \{0,1\}^n$ if $Y$ is a solution of the system. $\Phi$ is \textit{satisfiable} if it is satisfied by some assignment $U$.  
\end{definition}

Though the satisfiability of an instance of $k$-XOR-SAT can easily be determined in polynomial time by Gaussian elimination, we will not be concerned with 
this algorithm. Instead our aim is to analyze a local search procedure, called $\mathit{WalkSAT}$ \cite{papadimitriou1991selecting},  displayed as Algorithm~\ref{alg}, and originally investigated on random instances of $k$-SAT. Though possible in principle in several cases (e.g. \cite{schoning1999probabilistic,alekhnovich2006linear,coja2009smoothed,zhou2013exponential,coja2012analyzing-journal}) and well-understood from the standpoint of Statistical Mechanics \cite{semerjian2003relaxation,semerjian2004study}, 
 such an analysis is still quite complicated in general. 
 
Analyzing $\mathit{WalkSAT}$ on instances of $k$-XOR-SAT (rather than $k$-SAT) is motivated by the empirical observation that "curiously" \cite{guidetti2011walksat,achlioptas2015solution} XOR-SAT instances prove even harder for $\mathit{WalkSAT}$ than those arising from $k$-SAT. On the other hand, one may hope that obtaining a rigorous analysis of $\mathit{WalkSAT}$ may prove more tractable for the better understood problem $k$-XOR-SAT. While previous (highly nontrivial) such analyses concentrated on random instances \cite{barthel2003solving, semerjian2003relaxation,altarelli2008relationship},  we show that {\bf one can in fact obtain rigorous upper bounds on the expected running time of $\mathit{WalkSAT}$ on \textit{individual solvable instances}} of $k$-XOR-SAT, expressed in terms of {\bf (measurable) structural parameters of these
individual instances.}\footnote{we don't mean by this statement the fact that the expected running time may be easily predictable: these structural parameters may be themselves hard to compute.} We believe that such individual characterizations are important, as they make more transparent the structural properties of the input formula that influence the tractability of algorithms and heuristics. 

\begin{figure*}
\begin{center} 
\begin{pseudocode}[shadowbox]{Algorithm WalkSAT}{\Phi}
\mbox{Start with  assignment }U\mbox{ chosen uniformly at random.}\\
\mbox{while (there exist equations in $\Phi$ violated by U)}\\
\hspace{5mm}\mbox{ pick a random violated equation }C.\\
\hspace{5mm}\mbox{ flip the value of a random variable appearing in }C.\\
\mbox{\bf return}\mbox{ assignment }U.
\label{alg}
\end{pseudocode}
\end{center}
\caption{\textbf{Algorithm WalkSAT.}}
\end{figure*}

First of all, the following simple observation is true: 
on a satisfiable instance of $k$-XOR-SAT, one can reach a solution of the formula by legal moves of $\mathit{WalkSAT}$ from an arbitrary initial assignment. The easy formal result (Theorem~\ref{rec-xor}) is explicitly stated and proved in the Appendix. 

Given this observation, the following is a fairly natural research question: 
{\bf Given satisfiable formula $\Phi$ and initial assignment $X$, estimate $E[T_{\mathit{WalkSAT}}(\Phi,X)]$, the average number of steps $\mathit{WalkSAT}$ makes on $\Phi$ starting from $X$ in order to find a satisfying assignment.}

We will answer this question for a large-class of $k$-XOR-SAT instances by reducing the problem to one amenable to drift analysis \cite{he2001drift}. There will be two reductions: one (involving the formula hypergraph) works for uniquely satisfiable instances only. The second one reinterprets the dynamics of $\mathit{WalkSAT}$ by duality and uses a "dual" of the formula hypergraph. In this second setting, instances for which drift analysis leads to polynomial time upper bounds will require a mild additional "acyclicity" condition. 

\subsection{Other applications}
\label{sec:apps} 
The problem we are studying can be recast in contexts that are completely different from that of analyzing WalkSAT. Thus our results apply to these problems as well. We do not provide full details about such applications, only some bibliographic pointers to the relevant literature: 
\begin{itemize} 
\item[-] first, an application comes from the physics of complex systems and is given by the following dynamics, first investigated by Antal et al. \cite{redner-balance}: 

\begin{definition} \label{nondet}
 {\bf Constrained Triadic Dynamics} (CTD):  Start with 
graph $G=(V,E)$ whose edges are labeled $0/1$. A
triangle $T$ in $G$ is called \textit{balanced} if the
sum of its edge labels is 0 (mod 2). At any
step $t$, we randomly chose an imbalanced triangle $T$ and change the sign of a 
random edge of $T$ (thus
balancing $T$). The move might, however, unbalance other
triangles.
\end{definition}

The Constrained Triadic Dynamics  can be modeled by the $\mathit{WalkSAT}$ algorithm on an instance of 3-XORSAT \cite{triad-xorsat}. As further shown in \cite{istrate-balance}, one can sometimes analyze it using duality. So the results in this paper also provide bounds on the expected convergence time of CTD, on more general networks than those investigated in \cite{istrate-balance}. 

\item[-] Second, the particle systems studied in this paper are also related to certain \textit{lights-out games} \cite{sutner-lightsout}. See  \cite{istrate2015reachability} for more details. 
\end{itemize} 

\section{Preliminaries} 
\label{prelim}

In this section we briefly review the main notions we will need in the sequel. 

We will define instances of $k$-XOR-SAT from $k$-uniform hypergraphs, and conversely, we will associate $k$-uniform hypergraphs to instances of $k$-XOR-SAT. We will allow hypergraphs with \textit{self-loops}, i.e. hyperedges $e$ with $|e|=1$. We will even allow multiple self-loops to the same vertex. 

First, we introduce a more uniform notation for linear systems derived from hypergraphs: 

\begin{definition} 
Let $n\geq 1$, let $H$ be a hypergraph with $n$ vertices and let $Z_{n}=(Z_{1,n},Z_{2,n},\ldots, Z_{n,n})\in \{0,1\}^{n}$ be a boolean vector. We denote by  $H(Z_{n})$ the linear system with $n$ boolean variables $X_{1},X_{2},\ldots, X_{n}$ and equations
$\sum_{i\in e} X_{i}=\sum_{i\in e} Z_{i,n}$, 
where $e$ ranges over all hyperedges of $H$. 
\end{definition} 

By design $H(Z_{n})$ has $Z_{n}$ among the solutions. 
When $H$ is $k$-uniform, $H(Z_{n})$ is an instance of $k$-XOR-SAT. 
In particular, we refer to $K_{5}(Z_{n})$ as the \textit{complete 5-uniform linear system}. The reason for this name above is obvious: the \textit{formula hypergraph of $K_{5}(Z_{n})$} (having variables as vertices and equations corresponding to hyperedges) is the complete 5-uniform hypergraph. 

Issues of reachability by moves of $\mathit{WalkSAT}$ will motivate the following: 

\begin{definition}
For every pair of boolean configurations
$w_{1},w_{2}:V(H)\rightarrow \{0,1\}$ on hypergraph $H$ we
define a system of boolean linear equations $H(w_{1},w_{2})$ as
follows: Define, for each hyperedge $e$ a variable $z_{e}$ with
values in $\{0,1\}$. For any vertex $v\in V(H)$ we define the
equation
$ \sum_{v\in e}z_{e}=w_{2}(v)-w_{1}(v).$
In the previous equation the difference on the right-hand side is
taken in $\{0,1\}$; also, we allow empty sums on the left side.
System $H(w_{1},w_{2})$ simply consists of all equations, for all $v\in V(H)$.
\end{definition}

Given formula $\Phi$ and $e$ an equation of $\Phi$, we will denote by $\Var(e)$ the set of variables appearing in $e$. 

A \textit{multiset} is an unordered container of items whose elements have a (positive) multiplicity. The \textit{disjoint union} of multisets $A$ and $B$, denoted $A\sqcup B$, is the multiset that adds up multiplicities of elements in $A$ and $B$. 

Given hypergraph $H=(V,E)$ and $v\in V$ we will denote by $N(v)$ its \textit{open neighborhood}, defined as the set $\{w\neq v\in V:\mbox{ }(\exists e\in E), \{v,w\}\subseteq e\}$ and by $N[v]=\{v\}\cup N(v)$ its \textit{closed neigborhood}. 

Given instance $\Phi$ of $k$-XOR-SAT, the \textit{formula graph of $\Phi$}, denoted by $H(\Phi)$, is the hypergraph having variables of $\Phi$ as nodes and hyperedges which correspond to equations in $\Phi$. If $\Phi$ is satisfiable then $H(\Phi)$ will be a simple hypergraph, since for every $S\subseteq V$, $\Phi$ can contain at most one equation involving precisely the variables in $S$. 

We will deal with several discrete dynamical systems on hypergraphs. Consider such a dynamical system $\mathcal{D}$ on hypergraph $H$. Given two configurations 
$C_{1},C_{2}$  of $\mathcal{D}$, we will use notation $C_{1}\vdash C_{2}$ to denote the fact that $C_2$ follows from $C_1$ in one step of the dynamical system. 

A natural restriction on instances of $k$-XOR-SAT is the following: 

\begin{definition} A satisfiable instance $\Phi$ of $k$-XOR-SAT is \textit{connected} iff it cannot be partitioned into two non-empty instances of $k$-XOR-SAT, $\Phi=\Phi_{1}\cup \Phi_{2}$, 
with $\Phi_{1},\Phi_{2}$ having disjoint sets of variables. 
\end{definition} 

It is reasonable to require that instances of XOR-SAT we want to solve are connected. Indeed, the following results connects the expected running time of $\mathit{WalkSAT}$ on a general formula to computing the expecting running times on the connected components of the formula: 

\begin{lemma} Let $\Phi$ be a formula, and let $\Phi_{1},\ldots, \Phi_{r}$ be the subformulas of $\Phi$ corresponding to decomposing the formula hypergraph into connected components. Then: 
\begin{itemize} 
\item[-] Running $\mathit{WalkSAT}$ on $\Phi$ is equivalent to running $\mathit{\mathit{WalkSAT}}$ on $\Phi_{1},\ldots, \Phi_{r}$ independently. 
\item[-] We have: 
\[
E[T_{\mathit{WalkSAT}}[\Phi]]=\sum_{i=1}^{r}E[T_{\mathit{WalkSAT}}[\Phi_{i}]]
\]
\end{itemize} 
\end{lemma} 
\begin{proof} 
Formulas $\Phi_{1},\ldots, \Phi_{r}$ do not share variables. Because of this: 
\begin{itemize} 
\item[-] every step of $\mathit{\mathit{WalkSAT}}$ on $\Phi$ corresponds to one step of $\mathit{\mathit{WalkSAT}}$ on exactly one of $\Phi_{1},\ldots, \Phi_{r}$. 
\item[-] once a subformula $\Phi_{i}$ is satisfied, it will no longer be touched by $\mathit{\mathit{WalkSAT}}$. 
\end{itemize} 
The Lemma follows. 
\end{proof} 

The following restriction will appear in our results: 

\begin{definition} A satisfiable instance $\Phi$ of $k$-XOR-SAT is called \textit{cyclic} iff every variable appears in an even number of clauses (alternatively, if each equation $C$ of $\Phi$ is implied by the conjunction of all other equations of $\Phi$). A formula $\Phi$ is \textit{acyclic} iff no nonempty subformula of $\Phi$ (including $\Phi$ istself) is a cycle. 
\end{definition} 

A fairly important concept will turn out to be the following process (we call \textit{duality}) that associates a hypergraph to an instance of $k$-XOR-SAT:  

\begin{definition}
 Given instance $\Phi$ of $k$-XOR-SAT, the \textit{dual } $D(\Phi)$ of $\Phi$ is an undirected
hypergraph with self-loops
$D(\Phi)=(\overline{V},\overline{E})$, defined as
follows: $\overline{V}$ is the set of equations of
$\Phi$. Hyperedges in $D(\Phi)$ correspond to variables in $\Phi$
and connect all equations 
containing a given variable. In particular
we add a self-loop to an equation (vertex) $v$ if it contains a variable appearing \textit{only in 
$v$}. We may even add multiple self-loops to the same vertex. In other words $D(\Phi)$ is simply the dual of the formula hypergraph of $\Phi$. 
\end{definition}

Note that if $\Phi$ is an instance of $k$-XOR-SAT then $D(\Phi)$ is a {\bf $k$-regular hypergraph} (i.e. every vertex has degree exactly $k$). An example of duality is displayed in Figure~\ref{dual1}. The border nodes, edges in the primal hypergraph correspond to variables of the formula, and triangles to equations; vertices of the dual correspond to equations as well. Self-loops in the dual correspond to variables in the primal formula appearing {\bf exactly} in one equation. 

\begin{example} 
\label{hnru}
Let $1\leq r\leq n$ and let $u:{{n}\choose {r}}\rightarrow \{0,1\}$, where by ${{n}\choose {r}}$ we have denoted the family of subsets of $\{1,2,\ldots, n\}$ having exactly $r$ elements. Define system $H(n,r,u)$ by equations: 
\begin{equation} 
\sum_{A\in {{n}\choose {r}}, A\ni i} X_{A}=\sum_{A\in {{n}\choose {r}}, A\ni i} u(A), (\forall i \in 1, \ldots, n).  
\end{equation} 
$H(n,r,u)$ is a satisfiable instance of ${{n}\choose {r}}$-XOR-SAT.  Its dual, 
$D(H(n,r,u))$, is isomorphic to the complete $r$-ary hypergraph $K(n,r)$. 
\end{example}

\begin{center} 
\begin{figure}

\includegraphics{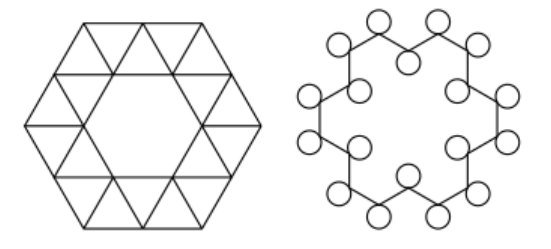}
\caption{The triadic cycle with 18 triangles, and its dual. }
\label{dual1}
\end{figure}
\end{center}

Odd-connected hypergraphs capture an important class of formulas: 

\begin{proposition} 
Instance $\Phi$ of $k$-XOR-SAT is acyclic iff the dual of its formula graph is odd-connected. 
\end{proposition} 
\begin{proof} 
Vertices of $D(\Phi)$ correspond to clauses of $\Phi$. An even cut of $D(\Phi)$ corresponds to a subformula $\Phi_{1}$ of $\Phi$ such that every variable in $\Phi$ appears in an even number of clauses of $\Phi_{1}$, i.e. to a cycle. 
\end{proof} 

We will use spectral measures for hypergraphs related to classical notion of 
the Cheeger constant of a connected graph \cite{chung1997spectral}: 
\[
    \phi(G) = \min_{0 < |A| \le \frac{n}{2} } \frac{|E(A,\overline{A})|}{|A|},
\]
where $E(A,\overline{A})$ is the \textit{edge boundary of set $A$}, that is, the set of edges with  one endpoint in set $A$ and another one in its complement. In practice a related version is often used: 
\begin{definition} 
The \textit{edge expansion (or Cheeger constant) $h(G)$ of graph $G$} is defined as
\begin{equation}
    h(G) = \min_{0 < |A| < n } \frac{n|E(A,\overline{A})|}{|A||\overline{A}|},
\end{equation} 
\end{definition} 

The analysis of voter model on graphs in \cite{aldous2013interacting,aldous-fill-book} required a different variant, related to the conductance of a graph and called in \cite{aldous-fill-book} \textit{Cheeger time}:  

\begin{definition} For a $k$-regular graph $G$ define {\rm the Cheeger time $\tau_{G}$ } of $G$ by  
\[
\tau_{G}=\max_{0<|A|<n}\frac{k|A||\overline{A}|}{n\cdot |E(A,\overline{A})|}\big(=\frac{k}{h(G)}\big)
\] 
\end{definition}

The definitions of Cheeger constant/time has been lifted up from graphs to simplicial complexes in several related but ultimately different ways: 
First, cohomological (coboundary) versions of Cheeger constant have appeared, implicitly or explicitly, in several works \cite{lm,gromov,dk,nr}. On the other hand Parzanchevski et al. \cite{parzanchevski2012isoperimetric} gave a 
combinatorial extension of the Cheeger constant to simplicial complexes. 

Our results will require yet another specially tailored Cheeger-like quantity, somewhat similar to the definition of coboundary expansion but with an easy combinatorial definition reminiscent of the Cheeger time: 

\begin{definition} Given a $k$-regular hypergraph $H$ define {\rm the odd Cheeger time $\tau_{H}$ } as 
\[
\tau_{odd}(H)=\max_{0<|A|\leq |V|}\frac{nk}{|E^{-}(A,\overline{A})|}. 
\] 
The definition above coincides, in the case of connected graphs, with the Cheeger time of $H$. 
\end{definition} 

We will use the following classical notions related to interacting particle systems. Though generally defined on infinite lattices \cite{liggett-ips}, they make sense on finite graphs as well (e.g. \cite{donnelly-welsh-interacting, cooper2013coalescing}): 

\begin{definition} 
Given a finite connected graph $G=(V,E)$, an \textit{annihilating random walk on $G$} is specified as follows: 
\begin{itemize} 
\item At every time step $t\geq 1$ the state of the system is specified by a \textit{value} function $f_{t}:V\rightarrow \{0,1\}$. We picture a $v\in V$ with $f_{t}(v)=1$ as a ball located on vertex $v$ at moment $t$. 
\item Given state $f_{t}$, the model proceeds, intutively, by choosing a random ball, located at some vertex $v$, and a random edge $(v,w)\in E$. Then the ball at $v$ moves to $w$. If it meets there another ball, the two balls annihilate each other. Formally $f_{t+1}(v)=0$, $f_{t+1}(w)=1-f_{t}(w)$, $f_{t+1}(z)=f_{t}(z)$ for $z\neq v,w$. 
\end{itemize} 
The quantity of interest is the annihilation time, defined as the smallest $t\geq 1$ such that $|v\in V:f_{t}(v)\neq \emptyset|\leq 1$. 
\end{definition}

\begin{example} 
Consider, for example, a triangle ABC initially having particles at vertices $A,B$. If particle at $A$ moves at $t=1$ at vertex $B$ then the two particles annihilate, and the annihilation time is 1. 
\end{example}  

A second, related process is the following: 

\begin{definition} 
Given a finite connected graph $G=(V,E)$, a \textit{coalescing random walk on $G$} is specified as follows: 
\begin{itemize} 
\item At every time step $t\geq 1$ the state of the system is specified by a \text{value} function $f_{t}:V\rightarrow \mathcal{P}(V)$. 
\item initially every vertex $v\in V$ holds either its own value or no value. That is, for all $v\in V$, either $f_{1}(v)=\{v\}$ or $f_1(v)=\emptyset$. 

\item Given state $f_{t}$, the model proceeds by choosing a random vertex $v$ with a nonempty cluster and a random edge $(v,w)\in E$. Then, intuitively,  the cluster of balls at $v$ moves to $w$. If it meets there another cluster, the two clusters merge. Formally $f_{t+1}(v)=\emptyset$, $f_{t+1}(w)=f_{t}(w)\cup f_{t}(v)$, $f_{t+1}(z)=f_{t}(z)$ for $z\neq v,w$. 
\end{itemize} 
The quantity of interest is \textit{the coalescence time}, defined as the smallest $t\geq 1$ such that $|v\in V:f_{t}(v)\neq \emptyset|\leq 1$. 


\end{definition} 

\begin{example} 
Consider, for example, a triangle ABC initially having particles at vertices $A,B$. We label the particles with the names of their vertices. If particle at $A$ moves at $t=1$ at vertex $B$ then the particles $A$ and $B$ coalesce, and vertex $B$ will subsequently hold a cluster formed by particles labeled $A$ and $B$. The coalescence time in this example is 1. 
\end{example}  

Next we give a formal treatment of the voter model on general graphs. Intuitively, all vertices in a graph hold a private initial opinion. Opinions evolve by having a vertex $v$ at a given moment copy the opinion of one of its neighbors $w$. There are multiple variants of the model, related to how $v,w$ are chosen (see e.g. \cite{berenbrink, cooper2016linear} for some variants and results). Formally: 

\begin{definition} 
Given a finite connected graph $G=(V,E)$, the \textit{voter model on $G$} is specified as follows: 
\begin{itemize} 
\item At every time step $t\geq 1$ the state of the system is specified by a \textit{value} function $f_{t}:V\rightarrow V$. Intuitively, every opinion is labeled, and $f_{t}(v)$ represents the opinion of vertex $v$ at time $t$. 
\item initially every vertex $v\in V$ holds its own value. That is, $f_{1}(v)=v$ for all $v\in V$. \item Given state $f_{t}$, the model proceeds by having one vertex $v$ assume the opinion of a neighbor $w$ of $v$. Formally $f_{t+1}(v)=f_{t}(w)$, $f_{t+1}(z)=f_{t}(z)$ for $z\neq v$. 
\end{itemize} 
The quantity of interest is the unanimity time, defined as the smallest $t\geq 1$ such that 
$|\{f_{t}(v):v\in V\}|=1$.  
\end{definition} 

\begin{example} 
Consider again a triangle $ABC$, with each vertex initially holding the opinion corresponding to their name. 
If 
\begin{itemize} 
\item[-] At time $t=1$ vertex $A$ copies the opinion of vertex $B$, and 
\item[-] At time $t=2$ vertex $C$ copies the opinion of vertex $A$
\end{itemize} 
then subsequently every vertex will hold opinion $B$. The unanimity time is $t=2$. 
\end{example} 

The voter-model consensus time is dual (see e.g. \cite{aldous-fill-book}) to the coalescence time of coalescing random walks. With distinct initial opinions at every vertex, the two times are equal in distribution, and in the standard graphical construction they are equal pathwise. In general, coalescence time is an upper bound.

We mention briefly in the sequel the notion of the \textit{$k$-core of a hypergraph $G=(V,E)$} (here $k\geq 2$ is an integer). Indeed, given such a $G$, its $k$-core is the hypergraph obtained from $G$ by iteratively removing vertices of degree less than $k$.

Finally, we will use the following simple result: 
\begin{lemma} 
Given random variable $X$ with support on $\mathbb{Z}_{+}$, $
E[X]=\sum_{i\geq 0} Pr[X> i].$
\label{avglemma}
\end{lemma}

\section{A Special Case With a Simpler Analysis: Unique $k$-XOR-SAT}
\label{motiv}
To motivate some of the concepts we will introduce in the sequel,  we first study a particular instance of the problem we study (that of analyzing the running time of $\mathit{WalkSAT}$ on individual instances of $k$-XOR-SAT) for which a simpler analysis (not involving interacting particle systems) is available: this is the case of \textit{instances of $k$-XOR-SAT having a single satisfying assignment.}  

To intuitively explain our approach to this case consider a hypergraph $H$ such that system $H(Z_{n})$ has exactly one solution $Z_{n}$. To any assignment $U_{t}$ considered at step $t$ by $\mathit{\mathit{WalkSAT}}$ one can associate a partition $(A_{t},\overline{A_{t}})$ of the variables $\{X_{1},X_{2},\ldots, X_{n}\}$ with \begin{equation} 
A_{t}=\{X_{i}: U_{t}(X_{i})\neq Z_{i}\}
\end{equation}
denoting the set of "bad variables".  We can (in principle) analyze the behavior of the $\mathit{WalkSAT}$ algorithm on $H(Z_{n})$ by employing the potential function $u(t)=|A_{t}|$. 

$\mathit{WalkSAT}$ evolves by flipping the value of a single variable. Therefore $u(t)$ can either decrease by 1 (if one "bad" variable becomes "good") or increase by one (if one "good" variables flips to "bad"). The following easy observation is crucial: 

\begin{lemma} 
For $t\geq 0$, equation $e$ is \textbf{not} satisfied by assignment $U_{t}$ iff
$|\Var(e)\cap A_{t}|\mbox{ is \textbf{odd.}}$
\end{lemma}
\begin{proof} If equation $e$ is specified by  
\[
\sum_{i\in Var(e)} x_i = a_{e}, 
\]
then, for $U_t$ to satisfy $e$ we must have 
\begin{align*} 
\sum_{i\in Var(e)} U_{t}(X_i)-\sum_{i\in Var(e)} Z_i = a_{e}-a_{e}=0. 
\end{align*} 
But the sum is nothing but 
\[
\sum_{i\in Var(e)} U_{t}(X_i)-Z_{i}=\sum_{i\in Var(e)\cap A_{t}}1. 
\]
for the sum to \textbf{not} be equal to 0 the cardinality of $Var(e)\cap A_t$ must be odd. 
\end{proof} 

This lemma motivates the following rather "exotic" notion of an odd cut in a hypergraph: 

\begin{definition}
Given hypergraph $H$ and partition $V(H)=A\cup \overline{A}, A\cap \overline{A}=\emptyset$, define 
\begin{itemize} 
\item[-] $OddCut(A)$ to be the subhypergraph of $H$ induced by 
edges $e$ such that {\bf $|e\cap A|$ is odd.} 
\item[-] $E^{-}(A,\overline{A})$ to be the set of pairs $(v,e)$ with $v\in A$ and $e\ni v$, $e\in OddCut(A)$. 
\item[-] $E^{+}(A,\overline{A})$ to be the set of pairs $(w,e)$ with $w\in \overline{A}$ and $e\ni w$, $e\in OddCut(A)$. 
\end{itemize} 
\label{expansion-hyper}
\end{definition} 

\begin{remark} 
In the definition of $OddCut(A),E^{-}(A,\overline{A})$ we allow (odd-size) hyperedges that may not contain a single vertex from $\overline{A}$ ! 

\end{remark}

The connection between these notions and the analysis of $\mathit{WalkSAT}$ is clear. Indeed, define $\Delta u(t)=u(t+1)-u(t)$. Then we have:  
\begin{itemize} 
\item[-] $\Delta u(t)=+1$, precisely when at step $t$ the chosen pair $(v,e)$ belongs to $E^{-}(A_{t},\overline{A_{t}})$. 
\item[-] Similarly, $\Delta u(t)=-1$, precisely when at step $t$ the chosen pair $(v,e)$ belongs to $E^{+}(A_{t},\overline{A_{t}})$. 
\item[-] Therefore  
\begin{equation} 
Pr[\Delta u(t)=1]=\frac{E^{-}(A_{t},\overline{A_{t}})}{E^{+}(A_{t},\overline{A_{t}})+E^{-}(A_{t},\overline{A_{t}})}
\end{equation} 
\begin{equation} 
Pr[\Delta u(t)=-1]=\frac{E^{+}(A_{t},\overline{A_{t}})}{E^{+}(A_{t},\overline{A_{t}})+E^{-}(A_{t},\overline{A_{t}})}, 
\end{equation} 
\end{itemize} 

In conclusion the expected value of quantity $u(t)$ can be controlled using the following combinatorial definition: 

\begin{definition} For a  hypergraph $H$ and set $A\subseteq V(H)$ define the 
\textit{odd  Cheeger drift $D_{odd}(A)$} as 
\begin{equation}
D_{odd}(A)=\frac{|E^{+}(A,\overline{A})|-|E^{-}(A,\overline{A})|}{|E^{+}(A,\overline{A})|+|E^{-}(A,\overline{A})|}.
\end{equation}
\end{definition}

\begin{example} For every $k$-regular \textit{graph} $G$ {\bf without self-loops}, the odd Cheeger drift of an arbitrary set $A$ is zero, as $|E^{+}(A,\overline{A})|=|E^{-}(A,\overline{A})|$ for all $A$. 
\label{drift-zero}
\end{example}

\begin{example} For every $k$-regular graph $G$ {\bf that may include self-loops}, the odd Cheeger drift of an arbitrary set $A$ is $\geq 0$. More precisely, if $L(A)$ is the multiset of self-loops of vertices in $A$, 
\[
D_{odd}(A)=\frac{|L(A)|}{2\cdot |OddCut(A)|-|L(A)|}
\]

Indeed, all edges $e\in OddCut(A)$ that are not self-loops contribute both to $E^{+}(A,\overline{A})$ and $E^{-}(A,\overline{A})$. On the other hand self-loops only contribute to $E^{-}(A,\overline{A})$. 
\end{example}

Note that the odd Cheeger drift $D_{odd}(A)$ is only well-defined for sets $A$ such that $OddCut(A)\neq \emptyset$.   
The characterization of all hypergraphs for which condition  $OddCut(A)\neq \emptyset$ is satisfied for all $A$ is related to \textit{parity domination in graphs} \cite{sutner-lightsout,amin1992neighborhood}, and is adapted to hypergraphs as follows:  



\begin{definition}
Given a connected hypergraph $H=(V,E)$, set of vertices $\emptyset \neq A\subseteq V$ is \textit{even parity dominating in $H$} if for every $e\in E$, 
$|A\cap e|$ is even.  $H$ is \textit{odd-connected} if it has no even parity dominating set $\emptyset \neq A\neq V$.


 
\end{definition} 

The introduced terminology allows us to characterize hypergraphs $H$ such that $D_{odd}(A)$ is well-defined for all $\emptyset \neq A\neq V$ by the following simple result: 

\begin{proposition} $OddCut(A)\neq \emptyset$ holds for all $\emptyset \neq A\neq V(H)$ \textit{iff} $H$ is odd-connected.  
\end{proposition}

Our main result in this section connects expected convergence time of WalkSAT to the properties of asymptotic drift. Intuitively, when the drift is positive (case (i).) or neutral (case (ii).) the convergence time is tractable. When the drift is negative (case (iii).) the convergence time is exponential. Formally: 
\begin{theorem}
\label{thm:3}
Given family $\Phi=(\Phi_{n})$ of connected, \textbf{uniquely satisfiable} instances of $k$-XOR-SAT, the following statements hold:  

\begin{itemize} 
\item[(i).] Suppose there exists $\delta >0$  such that for all nonempty sets $A\subset V$,  $D_{odd}(A)\geq \delta.$ Then 
\begin{equation}
\max_{X_{n}\in \{0,1\}^{n}} E[T_{\mathit{WalkSAT}}(\Phi_{n},X_{n})]\leq T(n),
\label{eq10}
\end{equation}

where $T(n)$ is a function such that 
\begin{equation} 
T(n)=O(n/\delta).
\end{equation}
\item[(ii).] Suppose case (i) does not apply but $D_{odd}(A)\geq 0$ for every nonempty set $A$. Then 

\begin{equation}
\max_{X_{n}\in \{0,1\}^{n}} E[T_{\mathit{WalkSAT}}(\Phi_{n},X_{n})]\leq T(n),
\label{eq10.1}
\end{equation}

where $T(n)$ is a function such that 
\begin{equation}
T(n)=O(n^2\cdot \tau_{odd}(H(\Phi_{n}))).
\end{equation}
\end{itemize} 
On the other hand, given family $\Phi=(\Phi_{n})_{n\geq 1}$ of connected instances of $k$-XOR-SAT, the following statement holds:  
 \begin{itemize} 
\item[(iii).] Suppose there exist $0<\eta_{1}<\eta_{2}< 1$ and $\delta >0$ such that for all  sets 
$A\subset V$ with $\eta_{1}n<|A|<\eta_{2}n$ we have $D_{odd}(A)<-\delta.$ Then for any sequence of assignments $(X_{n})_{n\geq 1}$ with $|X_{n}|\geq \eta_{2}n$ for all $n\geq 1$   
\begin{equation}
E[T_{\mathit{WalkSAT}}(\Phi_{n},X_{n})]\geq T_{H}(n),
\end{equation}

where $T_{H}(n)$ is a function such that
\begin{equation}
T_{H}(n)=\Omega(e^{\Omega(n)}).
\end{equation}
\end{itemize} 
\label{thm0}
\end{theorem} 
\begin{proof}

Case (i), a positive drift case, follows easily from classical results on drift analysis. For completeness we refer to the Additive Drift Theorem (see \cite{lengler2019drift}, pp. 95, or the Appendix, where it is cited verbatim). 
Given this result, the statement can be derived as follows: recall that $T_{WalkSAT}(\Phi_{n},X_{n})$ is the hitting time for the event $u(t)=0$. Because of positive drift $E[u(t)-u(t+1)]\geq \delta$. Since for all $X_{n}\in \{0,1\}^{n}$, $u(0)\leq n$, we have 
\[
E[T_{WalkSAT}(\Phi_{n},X_{n})]\leq \frac{n}{\delta}.
\]
So we can take $f(n)=n/\delta$.

Case (ii): Because of the condition we assumed, the stochastic process $u(t)$ is a supermartingale and has increments $\pm 1$. By the Doob decomposition we can stochastically dominate $u(t)$ by a martingale; in our particular case, we can take $Y_t$ to be simply a random walk with a reflecting barrier at $n$ such that that for every $0\leq k< n$, $E[\Delta Y_t]=0$, $ 
Pr[\Delta u(t)=0 | u(t)=k]=Pr[\Delta Y_{t}=0 | Y_{t}=k]$. 

By the definition of $Y_{t}$

\begin{align*}
& Pr[\Delta Y_{t}=-1 | Y_{t}=k]=Pr[\Delta Y_{t}=1 | Y_{t}=k]=  \frac{1}{2}(Pr[\Delta u(t)=-1 | u(t)=k]+\\ 
& + Pr[\Delta u(t)=+1 | u(t)=k])\geq \frac{1}{2}(Pr[\Delta u(t)=-1 | u(t)=k] 
\geq \frac{1}{2}\cdot \frac{1}{\tau_{odd}(H(\Phi_{n}))}.
\end{align*}

It is a consequence of the martingale drift theorem (see e.g. Corollary 26, \cite{kotzing-krejca}) that the expected time the process $Y_t$ hits one of values $0,n$ is $\Theta(n^2)$. However, what we want is the expected time to hit state 0 only, using the fact that $Y_t$ is a random walk with a reflecting barrier at $n$. 

Rather than attempting to derive this result from a general statement in the drift analysis literature, we give a direct argument: consider a random walk with a reflecting barrier at $n$ which moves to the left/right by one with probability $b$ and stays at the same position with probability $1-2b$ (for some $b<1/2$). 
 
 \begin{lemma}[Hitting time for a lazy reflected random walk]
Let $(X_t)_{t\geq 0}$ be a Markov chain on $\{0,1,\ldots,n\}$, where $0$ is absorbing and $n$ is reflecting. Suppose that for every $1\leq i<n$,
$\mathbb P(X_{t+1}=i-1\mid X_t=i)
=
\mathbb P(X_{t+1}=i+1\mid X_t=i)
=
p_i,$
with $p_i\geq b>0$, and $
\mathbb P(X_{t+1}=i\mid X_t=i)=1-2p_i.$
At the reflecting endpoint $n$, suppose that
$\mathbb P(X_{t+1}=n-1\mid X_t=n)\geq b$, 
with the remaining probability assigned to staying at $n$. Let $\tau_0=\inf\{t\geq 0:X_t=0\}.$
Then, for every starting point $x\in\{0,1,\ldots,n\}$,
\[
\mathbb E_x[\tau_0]\leq \frac{x(2n+1-x)}{2b}.
\]
In particular
\[
\max_{x} \mathbb E_x[\tau_0]=O\left(\frac{n^2}{b}\right).
\]
\label{mc-hit}
\end{lemma}

\begin{proof}
See the Appendix. 
\end{proof}
 
 Applying this result, we infer that 
 $T(n)$ can be taken to be $O(n^2\cdot \tau_{odd}(H(\Phi_{n})))$.  

Case (iii). 
A consequence of drift analysis. For completeness, recall that $u(t+1)-u(t)\in \{\pm 1\}$ (since flipping the value of a single variable can increase/decrease the value of the Hamming distance between the current assignment and the set of satisfying assignments by at most 1) and that the event that $\mathit{WalkSAT}$ finds a satisfying assignment is equivalent to the condition $u(t)=0$. 

Take now $n$ large enough such that $-1< D_{odd}(A)< -0.1$ for all $A$ with $0.1n<|A|< 0.2n$. Then choose e.g. $\epsilon = 0.1$ and apply the  Simplified Drift Theorem \cite{oliveto2011simplified} (see also the Appendix, where the corrected form is cited verbatim) inferring that there exist constants $c>1>d>0$ such that $Pr[T< c^{n}]< d^{n}.$ Therefore the expected time to hit zero is at least $c^{n}(1-d^{n})$. 

\end{proof}

\begin{figure}[t] 
\begin{center}
\includegraphics[width=8cm,height=6cm]{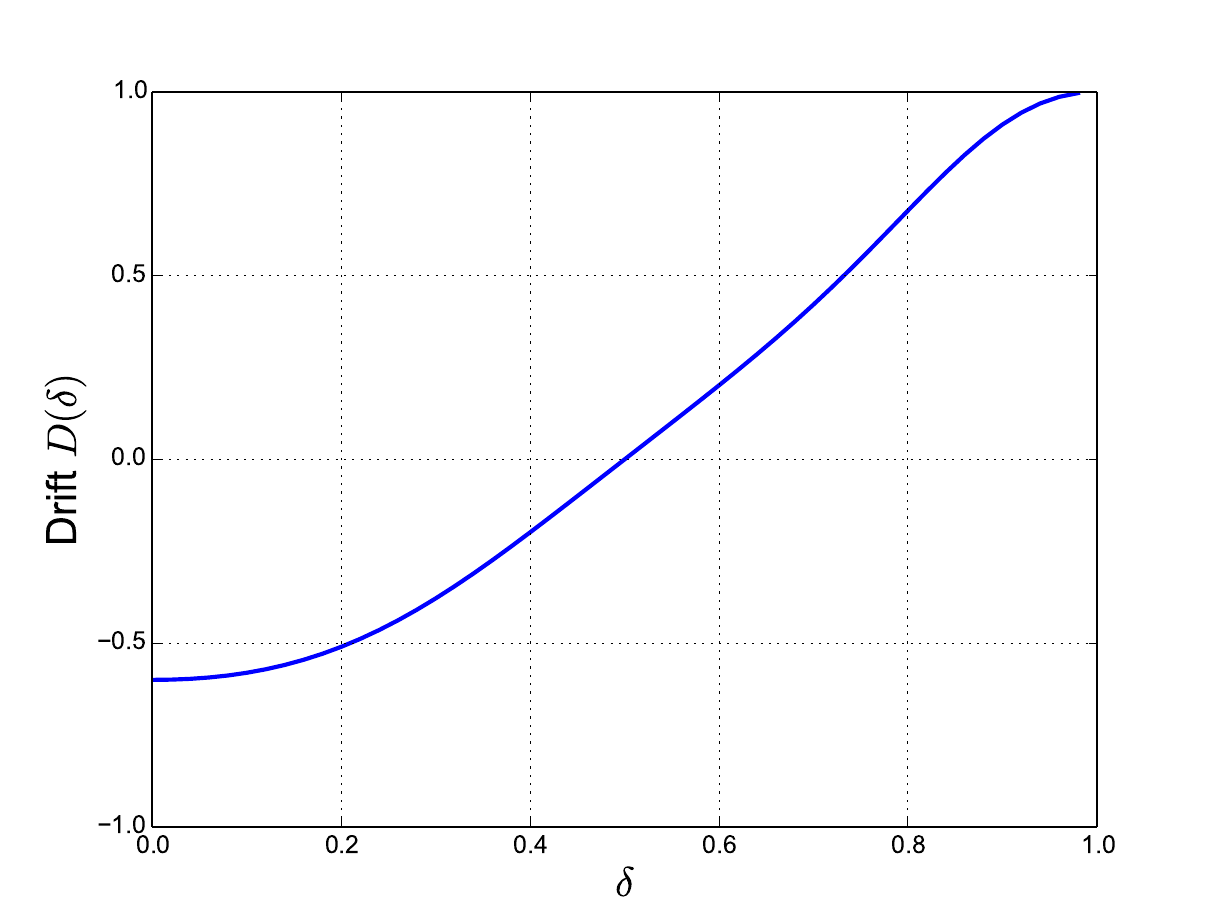}
\end{center} 
\caption{The (asymptotic) odd Cheeger drift of the complete 5-uniform hypergraph $K_{n,5}$.}
\label{k:5}
\end{figure}

Our next goal is to give an example to which  Theorem~\ref{thm0} (more precisely Case (iii)) applies. For this we first need the following easy observation: 
\begin{lemma} Given an arbitrary assignment $Z_{n}:=Z_{1,n}Z_{2,n}\ldots Z_{n,n}\in \{0,1\}^n$, 
$Z_{n}$ is the only solution of system $K_{5}(Z_{n})$. 
\label{1}
\end{lemma} 
\begin{proof} 
Recall that system $K_5(Z_n)$ consists of all equations of width 5
\[
\sum_{i\in S}X_i=\sum_{i\in S}Z_{i,n} \qquad (S\subseteq [n], |S|=5).
\]
Subtracting two equations that only differ by one variable ($X_{i}$ and $X_{j}$, respectively) 
we infer $X_{i}-X_{j}=Z_{i,n}-Z_{j,n}$ for all $1\leq i,j\leq n$. Setting $j=1$ we infer that the values of all variables $X_i$ are determined by setting the value of $X_{1}$: 
\[
X_i = X_1+Z_{i,n}-Z_{1,n}
\]
When $X_{1}=Z_{1,n}$ we obtain solution $X_i=Z_{i,n}$, that is vector $Z_n$. The alternative, setting $X_{1}=Z_{1,n}+1$, does not lead to a solution, because it corresponds to flipping all bits in $Z_{n}$, which is not a solution of the system (as all the equations have odd width).
\end{proof}

\begin{example}
Let $H=K_{n,5}$ be the complete $5$-uniform hypergraph with $n$ vertices \footnote{$K_{n,5}$ is, for every $Z_n\in \{0,1\}^n$, the formula graph of $K_{5}(Z_n)$.}.  
Then the odd Cheeger drift of arbitrary low-density subsets of $H$ will be negative (for large values of $n$). Indeed, if $|A|=\delta n$, the number of hyperedges $e$ containing 
\begin{itemize} 
\item[-] five vertices in $A$ is ${{\delta n}\choose {5}}\sim \frac{\delta^5}{5!}n^{5}$. Each vertex $v\in e$ will count for $E^{-}(A,\overline{A})$ in pair $(v,e)$
\item[-] three vertices in $A$ is 
${{\delta n}\choose {3}}{{(1-\delta) n}\choose {2}}\sim n^5 \delta^{3}(1-\delta)^{2}/12$. These three vertices will count for $E^{-}(A,\overline{A})$, the other two for $E^{+}(A,\overline{A})$. 
\item[-] 
exactly one vertex in $A$ is $\delta n{{(1-\delta) n}\choose {4}}\sim n^5 \delta(1-\delta)^{4}/24$. This vertex will count for $E^{-}(A,\overline{A})$, all the rest for $E^{+}(A,\overline{A})$.
\end{itemize} 

Thus, as $n\rightarrow \infty$ 
\[
|E^{-}(A,\overline{A})|= n^5[\frac{\delta^5}{24}+3\frac{\delta^3\cdot (1-\delta)^{2}}{12}+\frac{\delta (1-\delta)^{4}}{24}](1+o(1))
\]

On the other hand 
\[
|E^{+}(A,\overline{A})|= n^5[2\frac{\delta^3 (1-\delta)^{2}}{12}+4\frac{\delta(1-\delta)^{4}}{24}](1+o(1))
\]

"Asymptotic drift" quantity 
\begin{align*}
D_{odd}(\delta) & =
\lim_{n\rightarrow \infty, |A|=\delta n} \frac{|E^{-}(A,\overline{A})|-|E^{+}(A,\overline{A})|}{|E^{-}(A,\overline{A})|+|E^{+}(A,\overline{A})|}\\ 
& = \frac{\delta^{4}+2\delta^{2}(1-\delta)^{2}-3(1-\delta)^{4}}{\delta^{4}+10\delta^{2}(1-\delta)^{2}+5(1-\delta)^{4}}
\end{align*} 
 is plotted 
against density parameter $\delta$ in Figure~\ref{k:5}.  Note that the asymptotic drift is negative for $\delta < 1/2$. This yields (by our theorem) an exponential lower bound on the expected convergence time of WalkSAT. 
\label{2}
\end{example} 
\section{Beyond unique satisfiability}
\label{duality}
Theorem~\ref{thm0} in the previous section showed the relevance of concepts such as odd cuts and odd Cheeger drift in the analysis of the algorithm $\mathit{WalkSAT}$. However, its positive and neutral hypotheses are worst-case global conditions. They are likely to hold mainly for specially structured families, not generically for random k-XOR-SAT instances (the negative-drift condition may be closer to typical dense odd-k behavior). A further complication, that hampers the practical applicability of the result is the unique satisfiability restriction on formulas: it is not even clear that uniquely satisfiable formulas falling into Cases (i) and (ii). of the theorem exist! 

A difficulty in extending the analysis beyond the uniquely satisfiable case is the lack of a good analog of the progress measure $u(t)$: when the system has multiple solutions the set of variables can no longer be partitioned into good and bad ones. In such cases the formula may still have a backbone or spine \cite{aimath04}, consisting of variables that take the same value in all optimal assignments, but typically \cite{ibrahimi2012set,achlioptas2015solution} formulas  have "minibackbones" corresponding to local clusters of solutions, corresponding to variables whose values differ between the exponentially many different clusters but are fixed within a single cluster of solutions. Some (other) variables may be outside the 2-core of the formula hypergraph, playing no role in the satisfiability of the given formula (because one can always set such a variable to satisfy the unique constraint it participates in). However, the values of such variables are dependent on the values of variables in the 2-core, possibly impacting the dynamics of $\mathit{WalkSAT}$.

 In the rest of the paper we highlight a different route that works sometimes: rather than concentrating on \textit{variable-based} measures of progress, we will instead concentrate on the dynamics of  \textit{clause-based measures}. 

The main result below yields estimates of the expected convergence time of algorithm 
$\mathit{WalkSAT}$ on a satisfiable $k$-XOR-SAT instance  $\Phi$. Such bounds are similar to the ones in Theorem~\ref{thm0}, except that the relevant quantity turns out to be the odd Cheeger drift of the \textit{dual hypergraph $D(\Phi)$:}

\begin{theorem}

Given family $\Phi=(\Phi_{n})$ of  \textbf{acyclic} instances of $k$-XOR-SAT, the following statements hold:  

\begin{itemize} 
\item[(i).] Suppose there exists $\delta >0$  such that for all nonempty sets $A$,  $D_{odd}(A)\geq \delta.$ Then 
\begin{equation}
\max_{X_{n}\in \{0,1\}^{n}} E[T_{\mathit{WalkSAT}}(\Phi_{n},X_{n})]\leq T(n),
\label{eq1-1}
\end{equation}

where $T(n)$ is a function such that 
\begin{equation} 
T(n)=O(n/\delta).
\end{equation}
\item[(ii).] Suppose case 1 does not apply but $D_{odd}(A)\geq 0$ for every nonempty set $A$. Then inequality~(\ref{eq1}) holds with 
\begin{equation}
T(n)=O(n^2\cdot \tau_{odd}(D(\Phi_{n}))).
\end{equation}
\end{itemize} 
\label{cor}
\end{theorem} 

We illustrate the applicability of our new result as follows: 

\begin{example} Case (i). of Theorem~\ref{cor} is illustrated by instances whose formula hypergraph are the triadic cycle (Figure~\ref{dual1}). The dual $T_{3}(G)$ of this graph is depicted in Figure~\ref{dual1} as well. By the computation in Example 6 in Section~4, for every $A\subseteq V(T_{3}(G))$ 
\[
D_{odd}(A)=\frac{|A|}{2|Oddcut(A)|-|A|}\geq\frac{1}{3}. 
\]
Indeed, the last inequality is equivalent to $4 |A|\geq 2 |OddCut(A)|$ which is evident, since each node in $A$ has degee 3 and a self-loop, hence at most two adjacent edges in $OddCut(A)$. 

We infer that the convergence time of $\mathit{WalkSAT}$ on the corresponding XOR-SAT instance is linear. 
\end{example} 

\begin{center} 
\begin{figure}

\includegraphics{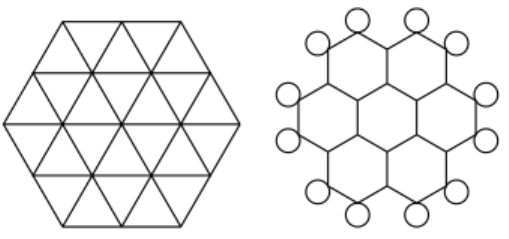}
\caption{(a). A section of the triangular lattice (b). its dual.}
\label{dual2}
\end{figure}
\end{center} 

\begin{example} 
An example of a family of formulas falling in the case (ii). of Theorem~\ref{cor} is based on a slight modification of an example of duality in Figure~\ref{dual2}, the hexagonal finite section of the triangular lattice. The modification works as follows: take \textbf{two} congruent copies of a hexagonal section of the triangular lattice, where the number of shells of hexagons, starting from the center and moving outward is $n$. For instance, the section drawn in Figure~\ref{dual2} (a) has $n=2$. Place each of the copies on a hemisphere, gluing the boundary edges to the equator. Glue the two hemispheres into a sphere, obtaining a graph $G_{n}$ (see Figure~\ref{fig:sphere1} for a picture of graph $G_2$). One can easily create a family of instances of 3-XOR-SAT whose formula graph is isomorphic to $G_{n}$: for each edge of $G_{n}$ consider value $u(e)\in \{0,1\}$. Define satisfiable instance $H(n,H_{6},u)$ of 3-XOR-SAT by adding for each triangle of $G_{n}$ composed of, say, edges $e_{1},e_{2},e_{3}$, equation 
\[
X_{e_{1}}+X_{e_{2}}+X_{e_{3}}= u_{e_{1}}+u_{e_{2}}+u_{e_{3}}. 
\]

\begin{figure} 
\includegraphics[height=40mm,width=40mm]{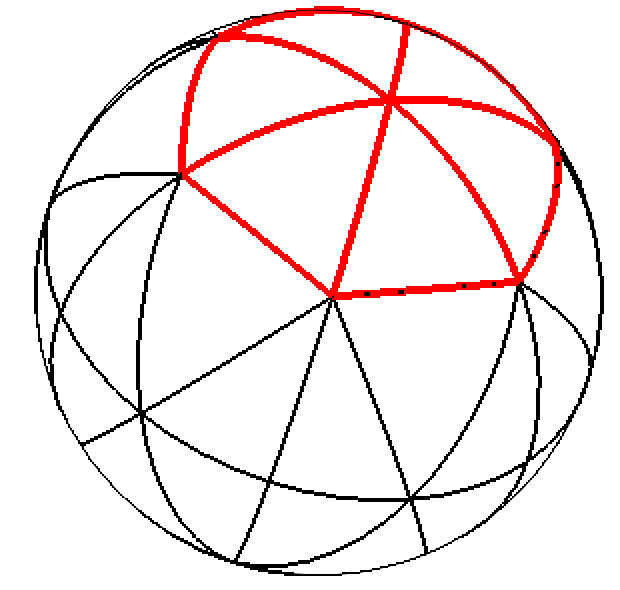}
\caption{Graph $G_2$. The north pole  of the sphere is at the intersection of the six red triangles (the sphere is slightly tilted for visibility).} 
\label{fig:sphere1}
\end{figure}

\begin{figure} 
\includegraphics[height=40mm,width=40mm]{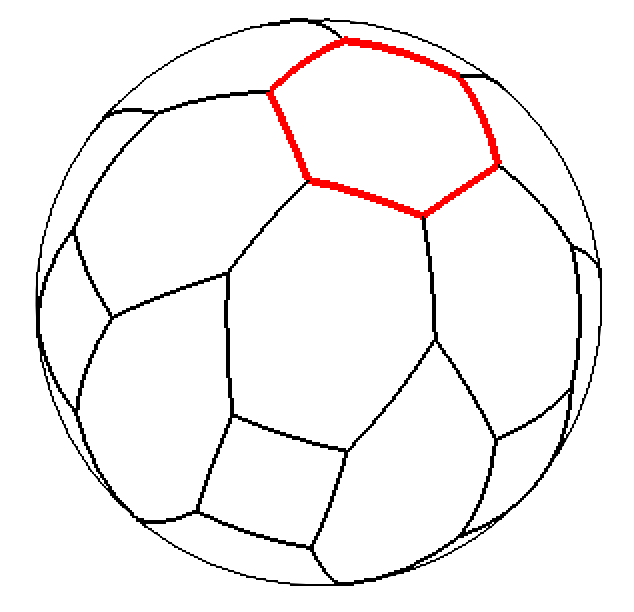}
\caption{The dual of $G_2$. The red hexagon corresponds to the six triangles surrounding the north pole in the representation of $G_2$.} 
\label{fig:sphere2}
\end{figure} 

The dual of this formula may be represented on the sphere as well, and is a graph without loops (hence falling, by Example~\ref{drift-zero}, into Case 2 of Theorem~\ref{cor}) consisting of regular hexagons, with the exception of the patches across the equator, where hexagons and rectangles alternate. This dual can be obtain by gluing the two copies of the dual in Figure~\ref{dual2} (b), where first we have eliminated the self-loops (since they correspond to the boundary edges - that get glued - of the triangular sections). See Figure~\ref{fig:sphere2} for a picture of the dual of graph $G_2$.  

We don't know how to compute the Cheeger time of this class of examples, though. 
\end{example}

A few further comments on our result are in order:  
\begin{itemize}
\item[-] The bounds on the expected running time of $\mathit{WalkSAT}$ will follow from corresponding bounds for hypergraph analogues of interacting particle systems, specifically annihilating and coalescing random walks, as well as of the voter model. Indeed:  
\begin{itemize} 
\item[-] We first note that for satisfiable instances of $k$-XOR-SAT, the reinterpretation of WalkSAT by duality corresponds to a dynamics that will motivate our version of annihilating random walk on hypergraphs. 
\item[-] in Definition~\ref{annrw-def} we indeed extend the classical definition of annihilating random walks to hypergraphs. Through an analysis parallels to a certain extent the one (for graphs) from \cite{aldous2013interacting} we obtain (Theorem~\ref{vm}) upper bounds on the mixing time of our generalization of annihilating random walks and, ultimately, of $\mathit{WalkSAT}$.  
\item[-]  
Unlike their graph counterparts, the ordinary  variants of random walks on hypergraphs \cite{cooper2013cover} or the recently defined \textit{$s$-walks} \cite{lu2012loose}, these models will be \textit{explosive}: the number of "particles" may in generally increase at a given step and be unbounded (for  coalescing random walks and the voter model). Such properties make these models qualitatively different from their graph versions.   

\item[-] Our result can be restated as the claim that the odd Cheeger drift dictates the nature of the convergence time of 
the multiset voter model (and, with it, of the application to algorithm $\mathit{WalkSAT}$). Of course, drift-based methods are well established in the analysis of local search heuristics \cite{he2001drift,droste2002analysis,lehre2014concentrated}. Our contribution is, therefore, \textbf{the identification of the odd Cheeger drift as the relevant quantity driving the dynamics of} 
$\mathit{WalkSAT.}$
\item[-] Conditions in Cases (i), (ii) are, of course, not exhaustive. They are chosen by analogy with the (first two) possible situations in drift analysis (and, more generally, in that of mixing in Markov chains) with tractable/intractable behavior arising from the connectedness of the configuration space (expansion and the existence of a bottleneck cut, respectively) 

\end{itemize} 
\end{itemize}

\section{Proof of Theorem~\ref{cor}}

\subsection{Reinterpreting $\mathit{WalkSAT}$ by duality}

The starting point of our analysis is the translation by duality of $\mathit{WalkSAT}$:

\begin{definition} 
For any instance $\Phi$ of $k$-XOR-SAT let $Z$ be an assignment to the variables of $\Phi$. Let $C_{Z}$ be the configuration on $D_{\Phi}$ (called \textit{the configuration dual to  $Z$}) defined as follows: a vertex $v$ has label 1 in $C_{Z}$ if and only if the corresponding equation is satisfied by $Z$.  
\end{definition} 

Given this definition, we have: 

\begin{theorem} 
For any instance $\Phi$ of $k$-XOR-SAT let $X_{0}$ be an initial assignment to the variables of $\Phi$. Let $C(X_{0})$ be the configuration dual to $X_{0}$. Suppose the algorithm $\mathit{WalkSAT}$ on $\Phi$ with initial assignment $X_{0}$ changes variable $x$ in (unsatisfied) clause $C$, resulting in assignment $X_{1}$. Then the configuration $C(X_{1})$ dual to $X_{1}$ is obtained by flipping the values of those nodes in hyperedge $x$ of $D_{\Phi}$ (which contains node $C$ whose initial value was 1).  
\end{theorem} 
\begin{proof} 
By changing the value of variable $x$ any equation that contains $x$ and was satisfied by $X_{0}$ becomes unsatisfied by $X_{1}$ and viceversa. On the dual this reads as follows: every vertex of the hyperedge that corresponds to variable $v$ changes value. 
\end{proof} 

When we translate by duality the behavior of $\mathit{WalkSAT}$ we get an instance of 
the following dynamics. It is this dynamics that we will analyze in the sequel:   

\begin{definition}
Let $H=(V,E)$ be a connected hypergraph. Define an \textit{annihilating random walk on $H$} (Figure~\ref{annrw}) by the following:
\begin{itemize} 
\item[(a).]
 {\bf Initial state:} Initially: $A_{i}\in  \{0,1\}$. We identify this configuration with  $\mathcal{B}=\{i\in V:A_{i}=1\}$, and call such a vertex $i$ \textit{live}. 
 \item[(b).]
 {\bf Moves:} Choose pair $i,e$ consisting of a {\bf random live node} $i$ and a random hyperedge $e=(i,j_{1},\ldots, j_{k-1})$ containing $i$. Simultaneously set $A_{v}=A_{v}\oplus A_{i}$ for all $v\in e$ (including $v=i$, which will result in $A_{i}=0$).  
 \item[(c).] {\bf Annihilation:} The event we want to measure is \textit{annihilation}, defined by condition $A_{v}=0$ for all  $v$. 
 \item[(c).] {\bf Stabilizing configurations:} To be able to talk about the expected time to annihilation we will limit ourselves to \textit{stabilizing configurations}, i.e. configurations $C$ of the annihilating random walk such that the annihilating configuration \textbf{0} is reachable from $C$ and all its descendants. We will denote the set of stabilizing configurations of $G$ by $\mathcal{S}(G)$. As implicitly shown in Theorem~\ref{rec-xor}, configurations $\mathcal{B}$ obtained by duality from assignments to satisfiable instances of XOR-SAT are indeed stabilizing. 
  \end{itemize} 
 \label{annrw-def}
 \end{definition} 

\begin{figure}
\begin{minipage}{.37\textwidth}
\includegraphics[height=45mm,width=45mm]{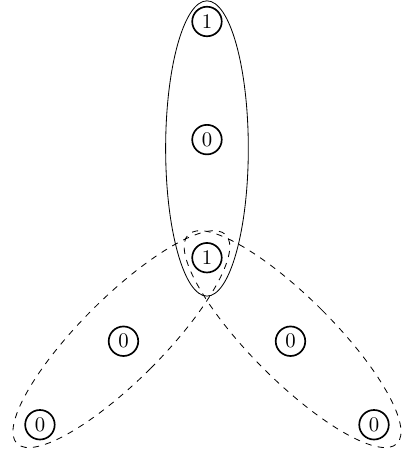}
\end{minipage} 
\begin{minipage}{.1\textwidth}
\Huge{
\[
\vdash
\]}
\end{minipage} 
\begin{minipage}{.37\textwidth}
\includegraphics[height=45mm,width=45mm]{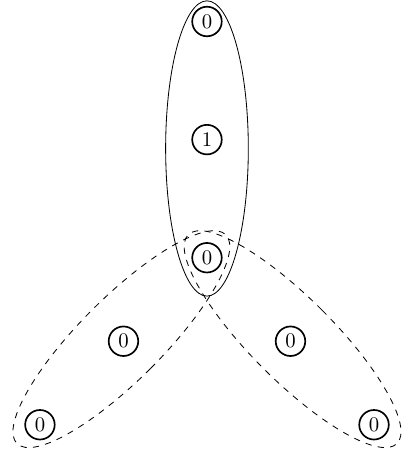}
\end{minipage} 
  
\caption{Annihilating random walk on hypergraphs.} 
\label{annrw} 
\end{figure}

We next define the annihilation time for annihilating random walks on hypergraphs: 

\begin{definition}
Given set of vertices $\mathcal{B}\subseteq V(G)$, \textit{$c_{ann}(G,\mathcal{B})$} is the minimum $t\geq 0$ such that in the a.r.w on $G$ started with $A_{i}=1$ if $i\in \mathcal{B}$, $A_{i}=0$ otherwise, at time $t$ we have $A_{i}= 0$ for all  $i$.  
\end{definition}

Theorem~\ref{cor} will follow from the following result, which bounds the convergence time of the annihilating random walk on a connected hypergraph:

\begin{theorem}
 Given family $H=(H_{n})$ of odd-connected hypergraphs, the following statements hold:

\begin{itemize} 
\item[(i).] Suppose there exists constant $\delta >0$  such that for all nonempty sets $A$,  $D_{odd}(A)\geq \delta.$ Then 
\begin{equation}
\max_{\mathcal{B}\in \mathcal{S}(H_{n})} E[c_{ann}(H_{n}, \mathcal{B})]\leq T_{H}(n),
\label{eq1}
\end{equation}

where $T_{H}(n)$ is a function such that 
\begin{equation} 
T_{H}(n)=O(n).
\end{equation}
\item[(ii).] Suppose case 1 does not apply but $D_{odd}(A)\geq 0$ for every nonempty set $A$. Then 
\begin{equation}
\max_{\mathcal{B}\in \mathcal{S}(H_{n})} E[c_{ann}(H_{n}, \mathcal{B})]\leq T_{H}(n),
\label{eq1.1}
\end{equation}
with 
\begin{equation}
T_{H}(n)=O(n^2\cdot \tau_{odd}(H_{n})).
\end{equation}
\end{itemize} 
\label{vm}
\end{theorem} 

In order to prove Theorem~\ref{vm}, it will prove more convenient to analyze annihilating random walks on $k$-uniform hypergraphs with one additional twist: we will study the \textit{lazy} version of annihilating random walk, the one in which the choice of node $i$ is not restricted to live nodes only.  More precisely, moves are specified as follows: 
\begin{itemize} 
\item[-] Choose random node $i$ and random edge $(i,j_{1},\ldots, j_{k-1})$ containing $i$. \item[-] Simultaneously set $A_{v}=A_{v}\oplus A_{i}$ for all $v\in e$ (including $v=i$, which will result in $A_{i}=0$). 
\end{itemize}

Making the annihilating random walk lazy increases, of course, the annihilation time, thus providing an upper bound on the convergence time of the $\mathit{WalkSAT}$ algorithm in the dual model. 

\section{Proof of Theorem~\ref{vm}} 
\label{sec:6}

The plan of the proof can be described as follows: 
\begin{itemize} 
\item[-] We will control annihilation using an extension of \textit{coalescing random walks} to hypergraphs. Since our focus is on obtaining upper bounds, one can couple (see Section~\ref{sec:7:3}) our generalizations of annihilating and coalescing random walks, and then prove upper bounds for this latter system instead. 

\item[-] Coalescing random walks are dual to the voter model. We define a \textit{the multiset voter model} (formally defined in Definition~\ref{def-multiset} below), and extend (Section~\ref{sec:7:4}) the classical duality \cite{liggett-ips,aldous-fill-book} between coalescing random walks and voter model to their new (explosive) versions. 
\item[-] One way to analyze the voter model on graphs is to couple it with a two-party voter model. Similarly, in Section~\ref{sec:7:5} we introduce a two-party version of the multiset voter model, study (Section~\ref{sec:7:6}) its reachability and recurrence property, then finally analyze it in terms of a "Cheeger-time" parameter of the underlying hypergraph. 

\item[-] Intuitively, such an analysis considers the evolution in time of a parameter $N_{t}$ of this model\footnote{defined the number of vertices having an odd number of opinions in the multiset voter model (or, equivalently of the coupled two-state voter model); see the definitions  below}. In one update step $N_{t}$ can only go up/down by 1 or stay the same. For odd-connected hypergraphs $H$, state ${\bf 0}$ is the only absorbing state, and what we want to bound is the expected time of the event $N(t)=0$.   

\item[-] In case (i). of the Theorem the walk $N(t)$ has a positive bias towards \textbf{0}. The convergence time is therefore linear. In case (ii). the random walk is at worst unbiased. It will hit zero only as a result of diffusive behavior. The convergence time would be quadratic if the probability of moving to the left/right would be (lower bounded by) a constant. It is not in general, thus we need to slightly alter the result to take into account this phenomenon. The Cheeger constant appears in the final upper bound in a similar way to that of Theorem~\ref{thm:3}. 
 \end{itemize}

\subsection{Annihilating random walks on hypergraphs: reachability, recurrence and stabilization} 
\label{reach}
If the hypergraph $H$ is actually a graph without loops, the long-term structure of configurations of the annihilating random walk is simple: either a single live particle survives (if $|V(H)|$ is odd) or none do. In the general case the nature of recurrent states may be more complicated: the number of live particles does {\bf not} always (weakly) decrease, as it does the case in the graph setting. There may be, therefore, recurrent states different from ${\bf 0}$ and those states with a single live node. A necessary condition for reachability was given in \cite{istrate-balance}:

\begin{lemma} \label{if}
 If $w_{2}$ is reachable from $w_{1}$ then system $H(w_{1},w_{2})$ has a solution in ${\bf Z}_{2}$.
\end{lemma}

\begin{proof}
Let $P$ be a path from $w_{1}$ to $w_{2}$ and let $z_{e}$ be the number of times edge $e$ is used on path $P$ (mod 2). Then $(z_{e})_{e\in E}$ is a solution of system $H(w_{1},w_{2})$. Indeed, element $w(v)$ (viewed modulo 2) flips its value any time an edge containing $v$ is scheduled.

\end{proof}

In \cite{istrate-balance} a partial converse of Lemma~\ref{if} was claimed. As we show in the Appendix (Theorem~\ref{graph}) such a result is, however, {\bf not} true, not even in restricted settings. This motivates the following 

\begin{open} What is the complexity of deciding the following problem:  Given connected hypergraph $H$ and two configurations $w_{1},w_{2}$ of an annihilating random walk on $H$, is $w_{2}$ is reachable (recurrent) from $w_{1}$?
\end{open} 

Note, however, that in the case we are interested in, the one corresponding to the setting of a satisfiable XOR-SAT formula,  we have $w_{2}=\textbf{0}$ and 
we {\bf do} have a converse, which yields an easy characterization of stabilizing configurations for the annihilating random walk: 

\begin{theorem} 
Configuration $w_{1}$ is a stabilizing configuration for the annihilating random walk on $H$ \textbf{if and only if} the system $H(w_{1},\bf{0})$ has a solution in $\{0,1\}$. 
\end{theorem} 
\begin{proof} 

Necessity follows by Lemma~\ref{if}. 

Suppose now that $w_{1}$ is a configuration on hypergraph $H$ such that system $H(w_{1},0)$ has a solution $u\in \{0,1\}$. Thus 
$w_{1}[v]=\sum_{e\ni v} u(e)$. 

Consider the system $S$ naturally corresponding to $H$. That is, variables $x_{e}$ of $S$ correspond to hyperedges $e$ of $H$. Equations of $S$ are defined to be 
\[
\sum_{e \ni v} x_{e}= w_{1}[v] (=\sum_{e\ni v} u(e)), v\in V(H).
\]
$S$ is satisfiable (since $u$ is a solution), hence starting with any initial assignment $X_{1}$ (in particular $X_{1}\equiv u$) we will reach a solution $X_{2}$. But the configuration corresponding to $X_{1}$ is indeed $w_{1}$, and the configuration corresponding to $X_{2}$ is ${\bf 0}$. 

\end{proof} 

\subsection{Explosive random walks and interactive particle systems on hypergraphs} 
\label{sec:7}

Next we define an analogue of coalescing random walks for hypergraphs (Figure~\ref{two-graph-coal}). We will assume that particles are labeled with their initial position in the hypergraph. The definition employs multiset sums (a somewhat peculiar operation) because, when scheduled, a given vertex passes a copy of the given particle to each of its neighbors in the chosen hyperedge. As a consequence, the number of copies of a given particle will generally increase.  Hence, a given vertex may receive multiple copies of the same particle, and we will need multiplicities to represent this situation.  

To extend the coupling of annihilating and coalescing random walks to our explosive versions we will need a condition called parity. This encodes the fact that, intuitively,  particles annihilate "in pairs". In the coalescing random walk model particles do not dissappear, but coalesce. So the natural analog for a live particle (in an annihilating random walk) is a group of particles having odd size (in the corresponding coalescing random walk):  

\begin{definition} 
Let $H=(V,E)$ be a connected hypergraph. Each vertex holds a multi-set of labels $A_{i}$. 
Define a \textit{coalescing random walk on $H$} by the following: 
\begin{itemize}
\item[(a).] \textbf{Initial state:}
 $A_{i}\subseteq \{i\}$. Note that $\mathcal{B}:= A_{1}\cup A_{2}\cup \ldots \cup A_{n}\subseteq [n]$. We will call a vertex $i$ with $|A_{i}|=$ odd \textit{live}. 
\item[(b).] \textbf{Moves:} Given node $i$ and hyperedge $e=(i,j_{1},j_{2},\ldots, j_{k})$ 
updating pair $(i,e)$ 
 proceeds by making $A_{j_{r}}:=A_{j_{r}}\uplus A_{i}$, for $r=1,\ldots, k$ and $A_{i}=\emptyset$. Here $\uplus$ refers to the {\bf multiset sum}, i.e. union with multiplicities. Note that the move never destroys any label (always $A_{1}\cup A_{2}\cup \ldots \cup A_{n}=[n]$) but may make some indices $i$ satisfy  $|A_{i}|=$ even. 
\item[(c).] \textbf{Parity (coalescence) from $\mathcal{A}$ on $\mathcal{B}$:} Given sets of vertices $\mathcal{A}, \mathcal{B}\subseteq V(G)$, \textit{$c_{coal}(H, \mathcal{B})$} is the minimum $t\geq 0$ such that, if starting with $A_{v}=\{v\}$ when $v\in \mathcal{A}$, $A_{v}=\emptyset $ otherwise, at time $t$ $|A_{j}|$ is even for every  $j\in \mathcal{B}$. 
\end{itemize} 
\end{definition} 

\begin{figure}[ht]
\begin{minipage}{.37\textwidth}
\begin{center}
\includegraphics[width=4.5cm,height=4.5cm]{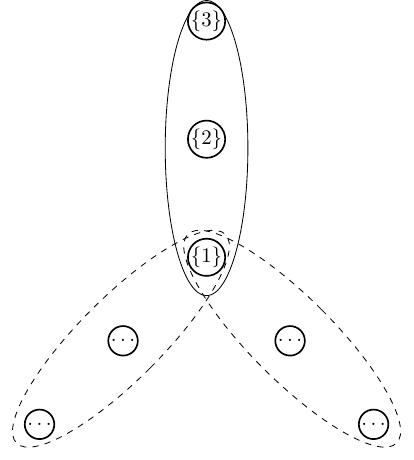}
\end{center}
\end{minipage}
\begin{minipage}{.1\textwidth}
\Huge{
\[
\vdash
\]}
\end{minipage} 
\begin{minipage}{.37\textwidth}
\begin{center}
\includegraphics[width=4.5cm,height=4.5cm]{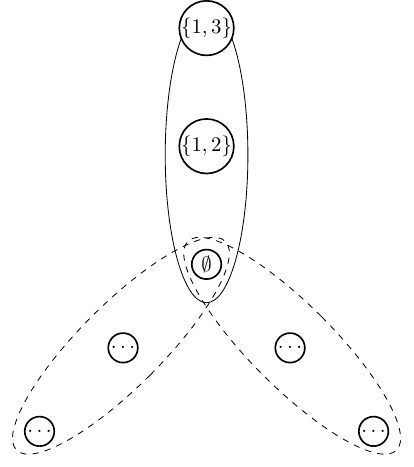}
\end{center}
\end{minipage}
\caption{One step of a coalescing random walk on hypergraphs.}
\label{two-graph-coal} 
\end{figure}

Finally, we present the "dual" to coalescing random walk, a multiset voter model (Figure~\ref{multiset}):  

\begin{definition} 
\label{def-multiset} 
Let $H=(V,E)$ be a connected hypergraph. Define the \textit{multiset voter model on $H$} by the following: 
\begin{itemize} 
\item[(a).]
{\bf Initial state:} $A_{i}=\{i\}$. Note that $A_{1}\cup A_{2}\cup \ldots \cup A_{n}=[n]$. 
\item[(b).]
{\bf Moves:} Given node $i$ and hyperedge $e=(i,j_{1},j_{2},\ldots, j_{k})$, 
updating pair $(i,e)$ results in setting $A_{i}=\uplus_{r=1}^{k} A_{j_{r}}$. Note that the operation may decrease the number of different "opinions" present in the system, if such opinions were only held by node $i$. 
\item[(c).]
{\bf Parity of opinions on $\mathcal{B}$ (from $\mathcal{A}$):} Given $\mathcal{A},\mathcal{B}\subseteq V(H)$, \textit{parity time $C_{\textit{VM}}(H,\mathcal{B},\mathcal{A})$} is the minimum time $t$ such that every initial opinion from $\mathcal{A}$ is present an even number of times (perhaps zero) among nodes in $\mathcal{B}$. We will omit the second argument when $\mathcal{A}=V$.
\end{itemize} 
\end{definition} 

\begin{figure}[ht]
\begin{center} 
\begin{minipage}{.37\textwidth}
\begin{center}
\includegraphics[width=4.5cm,height=4.5cm]{unreach123.pdf}
\end{center}
\end{minipage}
\begin{minipage}{.1\textwidth}
\Huge{
\[
\vdash
\]}
\end{minipage} 
\begin{minipage}{.37\textwidth}
\begin{center}
\includegraphics[width=4.5cm,height=4.5cm]{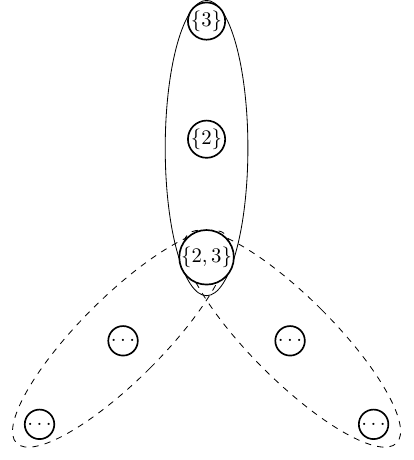}
\end{center}
\end{minipage}
\end{center} 
\caption{One step of the multiset voter model on hypergraphs.}
\label{multiset} 
\end{figure}

\subsection{Coupling annihilating and coalescing random walks} 
\label{sec:7:3}

We now extend a coupling argument valid in the case of graphs: 


\begin{theorem} Suppose $G$ is a connected hypergraph and $\mathcal{B}\subseteq V(G)$ is stabilizing. Then 
\begin{itemize} 
\item[-] in the coalescing random walk on $G$ starting from $\mathcal{B}$ one can reach coalescence. 
\item[-] one can couple the coalescing and annihilating random walks on $G$ such that $
 c_{ann}(G, \mathcal{B})= c_{coal}(G,\mathcal{B}).$
 
\item[-] One can also couple the \textit{lazy} versions of coalescing/annihilating random walks such that annihilation and coalescence time of the two processes coincide.  
\end{itemize} 
\label{coup-hypergraphs} 
\end{theorem} 
\begin{proof}  

We will define the following stochastic process $P$: 

\begin{enumerate} 
\item {\bf Initial state:} $A_{i}=\{(i,\infty)\}$ for $i\in \mathcal{B}$, $A_{i}=\emptyset$ otherwise. Note that $A_{1}\cup A_{2}\cup \ldots \cup A_{n}=\mathcal{B}\times \{\infty \} $ and that each $A_{i}$ contains at most one index $b_{i}$ with $(b_{i},\infty)\in A_{i}$. As the process progresses, if there exists $b_i$ such that $(b_i,\infty)\in A_i$, we will call the multiset $A_i$ \textit{live}, and call $b_{i}$ \textit{the witness} for $A_{i}$. Otherwise, we will call multiset $A_i$ \textit{dead}.  Also, define $B_{i}=A_{i}\setminus \{(i,\infty)\}$ if $i$ is live, $B_{i}=A_{i}$ otherwise. 
\item {\bf Move:} At time $t$: Choose {\rm random vertex $i$ (not necessarily live)}. Choose random edge $(i,j_{1},\ldots, j_{k})$. For $r=1,\ldots, k$ 
\begin{itemize} 
\item If both $A_{i},A_{j_{r}}$ are live then make $A_{j_{r}}= (B_{i}\sqcup B_{j_{r}})\sqcup \{(b_{i},t),(b_{j_{r}},t)\}$\footnote{The intuition behind this definition is that particles $b_i,b_j$ collide (and die) at time $t$. At time $t$, the set $A_v$ will contain the set of all particles (dead or live) located at $v$. Each such particle is labeled either with the time of its death, or with $\infty$, if it is still alive.}.  
\item If, on the other hand,  at most one of $A_{i},A_{j_{r}}$ is live then make $A_{j}:=A_{i}\sqcup A_{j_{r}}$. 
\end{itemize} 
Finally make $A_{i}=\emptyset$. Note that if we "move" a dead set $A_{i}$ to a live multiset $A_{j}$ then $A_{j}$ will still be live. 
\item {\bf Stopping:} \textit{Stopping time $c_{P}(G)$} is the minimum $t\geq 0$ such that at most one $i$ is live (exactly one if $n$ is odd, none if $n$ is even)
\end{enumerate} 

\begin{claim} The following are true: 
\begin{enumerate} 
\item $P$ {\bf observed on $[n]\times \{\infty\}$  and moves of live multisets only} (Figure~\ref{arw}) is the annihilating random walk on $G$ starting from configuration $\mathcal{B}$.  If $n$ is even then at time $c_{P}(G)$ all particles have annihilated. Consequently $c_{ann}(G, \mathcal{B})\leq c_{P}(G, \mathcal{B})$. 
\item $P$ where we disregard second components in all pairs Figure~\ref{crw}) is identical to the coalescing random walk on $G$ and
$c_{P}(G, \mathcal{B})=c_{coal}(G,\mathcal{B})$. 
\end{enumerate} 
\end{claim} 

The proof of the Claim has two cases, displayed in Figure~\ref{first}.  

In case (a), when observed on $[n]\infty$, particle $(a_i,\infty)$ (intuitively, a live particle at $a_i$ that hasn't died yet), moves to the hyperedge containing particle $(a_j,\infty)$. After this move the initial vertex will contain no particle, while the second vertex will contain both particles $a_i,a_j$, both with a death time (corresponding to their annihilation) $t$. 

In case (b) the particle $a_i$ at the first vertex moves to the second one, and still stays alive. 

In both cases the process, observed only on $[n]\times \infty$ is identical to the annihilating random walk. Note that if $n$ is even then when coalescence occurs in the coalescing random walk all particles have died in the annihilating random walk.  If $n$ is odd then when coalescence occurs in the coalescing random walk at most one live particle remains in the annihilating random walk.

\begin{figure}[t]
\begin{center}
\includegraphics[width=8cm]{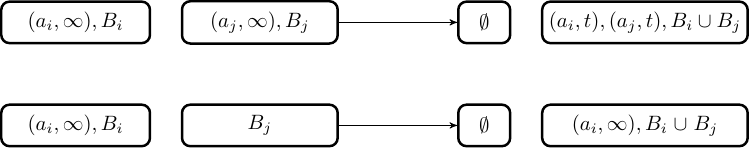}
\end{center}
\caption{The two cases of stochastic process $P$. On each side only two nodes lying inside a common hyperedge are pictured. In the first scenario, at time t particles $a_{i}$ and $a_{j}$ meet and annihilate.}
\label{first} 
\end{figure}

\begin{figure}[t]
\begin{center}
\includegraphics[width=8cm]{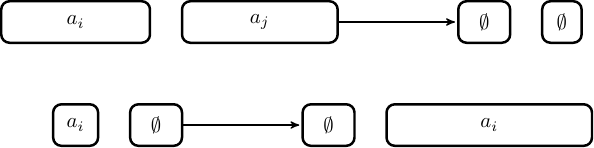}
\end{center}
\caption{First coupled version: annihilating random graphs (the two cases). Only two nodes inside a common hyperedge are pictured on each side.}
\label{arw}
\end{figure}

\begin{figure}
\begin{center}
\includegraphics[width=8cm]{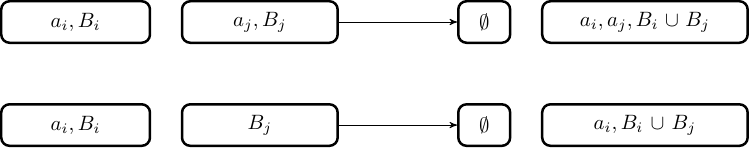}

\end{center}
\caption{Second coupled version: coalescing random walks (the two cases).  Only two nodes inside a common hyperedge are pictured on each side.}
\label{crw}
\end{figure}
\end{proof}

\subsection{Duality of coalescing random walks on hypergraphs and the multiset voter model}
\label{sec:7:4}
The reason a result such as Theorem~\ref{coup-hypergraphs} is interesting is that on graphs (see \cite{liggett-ips, aldous2013interacting, aldous-fill-book}) $c_{coal}(G)$ is identical (via duality) to the coalescence time of voter model $C_{\textit{VM}}(G)$, which can in turn be upper bounded in terms of a so-called \textit{Cheeger time of graph $G$}, essentially the inverse of the more well-known Cheeger constant of $G$. 

Similar results holds on hypergraphs, although we will need to give them in a slightly more general form: 

\begin{theorem} Let $H$ be a connected hypergraph and $\mathcal{B}$ be a stabilizing configuration on $H$. Then the following are true: 
\begin{itemize} 
\item[-] one can reach parity on $\mathcal{B}$ in the multiset voter model.
\item[-] the coalescence time $c_{coal}^{lazy}(H,\mathcal{B})$ of the lazy coalescing random walk and the parity time of the associated multiset voter model $C_{\textit{VM}}(H,\mathcal{B})$ are identically distributed.
\end{itemize}  
\label{dual}
\end{theorem} 
\begin{proof}
The proof is an adaptation of the classical duality argument \cite{liggett-ips} between (classical) coalescing random walks and the voter model: we will define a process on \textit{oriented hyperedges} in $H$ (that is edges with a distinguished vertex) that will be interpreted in two different ways: as parity in the multiset voter model and coalescence in the coalescing random walk. 

The process is depicted in Figure~\ref{cou}. In each of its two interpretations the basis of the process are events consisting of a hyperedge $e$ and a vertex $v$ of $e$. 
\begin{itemize} 
\item In the coalescing random walk (left picture) the pair $(v,e)$ triggers sending a copy of all the balls at $v$ to every other vertex of $e$. 
\item In the voter model (right picture) the pair $(v,e)$ means that vertex $v$ gets a copy of all opinions held by every vertex $w\in e$, $w\neq v$. These 
\end{itemize} 

In the figure we represent hyperedges vertically at the moment the given hyperedge event occurs (times $t_{1}$ and $t_{2}$ in the coalescing random walk). Horizontal lines (e.g. for ball 3 between moments $t_{1}$ and $t_{2}$) refer to histories not interrupted by any hyperedge event between the corresponding times. 

A \textit{left-right path $P$} between node $i$ and node $j$ is a sequence of hyperedge events and horizontal lines such that: 
\begin{itemize} 
\item $P$ starts with a horizontal line of node $i$ and ends with a horizontal line of node $j$. 
\item Every horizontal line of a node is followed by a hyperedge event with the corresponding node being pointed. 
\item Every hyperedge event is followed by an unique horizontal line corresponding to a {\bf non-pointed node}.
\end{itemize} 

For instance, in the picture from Figure~\ref{cou} we have represented three left-right paths, between node 2 and each of nodes 1,4,5. 

In the coalescing random walk the activation of hyperedge $e=[j\rightarrow i_{1},i_{2},\ldots i_{r}]$ pointed at $j$ is interpreted as vertex $j$ being chosen (together with edge $e$), thus sending a copy of its cluster of balls to all other neighbors. In the multiset voter model the activation of the same pointed hyperedge is interpreted as $j$ adopting the multiset union of opinions of $i_{1},i_{2},\ldots, i_{r}$. 


\begin{figure*}
\begin{center}
\includegraphics[width=6cm,height=4cm]{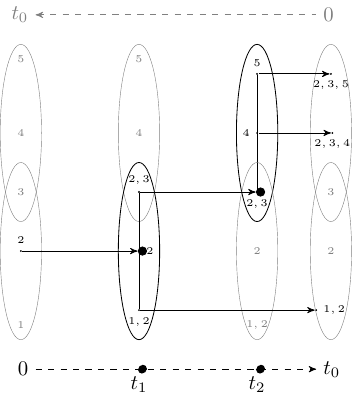}~~~~~\includegraphics[width=6cm,height=4cm]{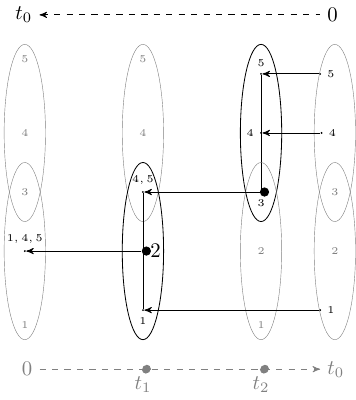}
\end{center}
\caption{Coupling the lazy coalescing random walk and the multiset voter model. Time runs from left to right in the lazy coalescing random walk and right to left in the multiset voter model. Digits represent node indices, unless bolded, in which case they represent balls/opinions. a) At time $t_{1}$ (in the lazy coalescing random walk) copies of balls at (pointed) node $2$ are sent to nodes 1 and 3. Similarly, at time $t_{2}$ copies of balls at (pointed) node 3 are sent to nodes 4 and 5. b) The corresponding multiset voter model. First $3$ aquires the opinions of $4$ and $5$; then $2$ aquires opinions $(4,5)$ of $3$ and the opinion of $1.$} 
\label{cou}
\end{figure*}

Just as in the ordinary coalescing random walk/voter model, the existence of a left-right path between nodes $i$ and $j$ (e.g. 
$(2,1),(2,4),(2,5)$) is interpreted as the event: 
\begin{itemize} 
\item "at $t_{0}$ node $j$ holds a ball with label $i$." (in the coalescing random walk). 
\item "at $t_{0}$ node $i$ holds opinion $j$ with multiplicity at least one." (in the multiset voter model) 
\end{itemize} 

\noindent Moreover, one path may contribute (when it does) with \textit{exactly one ball/opinion} of a given type. 

Consider now the event: "at $t_{0}$ every node in $\mathcal{B}$ on the right-hand side is connected to nodes on the left-hand side by an even number of paths". 

\begin{itemize} 
\item[-] In the coalescing random walk this is equivalent to "at $t_{0}$ we have parity from $\mathcal{B}$"
\item[-] In the multiset voter model this is equivalent to "at $t_{0}$ we have parity of opinions on $\mathcal{B}$"
\end{itemize} 

Since $\mathcal{B}$ is a stabilizing configuration for the coalescing random walk, there is an activation diagram that corresponds to coalescence in this model. Reintepreting this diagram by duality in the multiset voter model corresponds to parity of $\mathcal{B}$ in the multiset voter model. So the first statement of the theorem is true. 

The second point of the theorem is true since, just as in the classical case, the time in the two diagrams is the same (except it is interepreted in opposite directions) 
\end{proof}

\begin{example}
Suppose (for a different example) that in the lazy coalescing random walk in Figure~\ref{cou} we start with balls at node set $\mathcal{B}=\{2,4\}$, i.e. the initial state of the system at time 0 is $S_{0}=[\emptyset, \{2\},\emptyset, \{4\},\emptyset ]$. Then at time $t_{0}$ the system is in the state $S_{1}=[\{2\}, \emptyset, \emptyset, \{2,4\},\{2\} ]$. Nodes $1,5$ have an odd number of balls, while 2,3,4 have an even number. If we were to run the lazy annihilating random walk, nodes $1,5$ would be those that still have a particle, while $2,3,4$ don't. 

In the voter model (displayed with time going backwards) we start with state $V_{0}=[\{1\}, \{2\},\{3\},\{4\},\{5\} ]$ and end up with state $V_{1}=[\{1\}, \{1,4,5\},\{4,5\},\{4\},\{5\} ]$. Values 1,5 are those that have not yet reached parity on set $\mathcal{B}=\{2,4\}$, while 2,3,4 have. 
\end{example} 

\subsection{The two-party counterpart of the multiset voter model}
\label{sec:7:5}
Upper bounding the coalescence time of the classical voter model on graphs can be achieved \cite{donnelly-welsh-interacting,aldous2013interacting, aldous-fill-book} by coupling it with a "two-party" counterpart and analyzing this latter model instead. In the sequel we accomplish a similar task on general hypergraphs:  
\begin{definition}
Let $H=(V,E)$ be a connected hypergraph. Define the \textit{two-party voter model on $H$} by the following: 
\begin{itemize}
\item[(a). ]
{\bf Initial state:} $A_{i}\in \{0,1\}$ for all $i\in V$. We denote $\mathcal{A}=\{i\in V: A_{i}=1\}$. 
\item[(b). ]
{\bf Moves:} Given node $i$ and hyperedge $e=(i,j_{1},j_{2},\ldots, j_{k})$, 
updating pair $(i,e)$ results in setting $A_{i}=\oplus_{r=1}^{k} A_{j_{r}}$, keeping all values $A_{j_{r}}$ unchanged. 
\item[(c). ]
{\bf Parity of opinions on $\mathcal{B}$:} Given $\mathcal{A,B}\subseteq V(H)$, the \textit{parity time $c_{2-VM}(H;\mathcal{A},\mathcal{B})$} is the minimum time $t$ such that, starting from configuration $\mathcal{A}$, at time $t$ and subsequently $\oplus_{i\in \mathcal{B}} A_{i}=0.$
\end{itemize} 

Note that, unlike the multiset voter model, in the two-party voter model we allow initial states $\mathcal{A}$ where at time $t=0$ "some nodes do not hold any opinion". 
\end{definition}

\begin{theorem} Let $H$ be a connected hypergraph and $\mathcal{B}\subseteq V(H)$ be a stabilizing set. Then 
\begin{itemize} 
\item[-] for every set $\mathcal{A}\subseteq V(H)$, in the two-party voter model started from configuration $\mathcal{A}$  one can reach parity of opinions on $\mathcal{B}$.
\item[-] 
for every $\mathcal{A}\subseteq V(H)$ 
one can couple the multiset voter model and the two-state voter model with initial state $\mathcal{A}$ such that whenever we have parity on $\mathcal{B}$ in the multiset voter model we have parity of opinions on $\mathcal{B}$ from $\mathcal{A}$ in the two-party voter model.
\end{itemize}
\label{coup-voter} 
\end{theorem} 
\begin{proof} 

The first part follows by duality of the multiset voter model and the lazy coalescing random walk: by Theorem~\ref{coup-hypergraphs}, since $\mathcal{B}$ is stabilizing, one can reach parity 
from $\mathcal{B}$ in the coalescing random walk Therefore in the multiset voter model one can reach parity on $\mathcal{B}$ from $V$, hence from $\mathcal{A}$ as well. 

Now, given a run of the multiset voter model, define a (coupled) run of the two-state voter model with initial state $\mathcal{A}$ by defining, for every $i\in V$ and every moment $t$, $A_{i}$ to denote the parity of the multiset of opinions \textit{from set $\mathcal{A}$ only} held at moment $t$ by vertex $i$.

Since $\mathcal{B}$ is a stabilizing set one can reach parity of opinions on $\mathcal{B}$ in the voter model. At that time each opinion (including those in $\mathcal{A}$) is present an even number of times among nodes in $\mathcal{B}$. Therefore we have parity of opinions on $\mathcal{B}$ in the coupled two-state voter model as well. 

\end{proof} 

\begin{corollary} In the settings of Theorem~\ref{coup-voter} we have, for every $t\in {\bf N}$
\begin{align*}
Pr[C_{2-\textit{VM}}(H;\mathcal{A},\mathcal{B})>t]\leq Pr[C_{\textit{VM}}(H;\mathcal{B})>t] 
\end{align*}
and 
\begin{align*}
E[C_{2-\textit{VM}}(H;\mathcal{A},\mathcal{B})]\leq E[C_{\textit{VM}}(H;\mathcal{B})]
\label{2vm}
\end{align*}
\end{corollary} 
\begin{proof} 
The first inequality follows from coupling. For the second we apply Lemma~\ref{avglemma}.
\end{proof} 

We complement the second inequality in the previous corollary by one with the opposite direction, derived as follows: Consider the following process, parameterized by a positive number $\epsilon > 0$, which yields a random model we will call ${\bf D}_{\epsilon}$: 
\begin{itemize} 
\item[-] we partition the vertices of $H$ into two parts, $D$ and $\overline{D}$ by including each vertex into $D$ independently with probability $1/2-\epsilon$.    
\item[-] we run the two-state voter model from configuration $D$ (i.e.  0 on labels of vertices of $\overline{D}$ ("reds") and $1$ on vertices of $D$ ("blues")).
\item[-] We denote by $D_{t}$ the set of vertices labeled 1 at time $t$, by $N_{t}^{D_{\epsilon}}$ the cardinal of $D_{t}$,  and by $\Delta_{t}^{D_{\epsilon}}$ the difference in the number of ones as a result of the (possible) jump at time $t$. 
\item[-] Denote by $C^{D_{\epsilon}}$ the smallest time $t\geq 0$ when $N_{t}^{D_{\epsilon}}=0$. Denote by  $C_{2-\textit{VM}}^{\bf D_{\epsilon}}(H;\mathcal{B})$ the corresponding parity time. 
\end{itemize}

\begin{lemma} For every $0<\epsilon<1/2$ we have \begin{equation} 
E[C_{\textit{VM}}(H;\mathcal{B})]\leq \frac{2}{1-2\epsilon}\cdot E[C^{\bf D_{\epsilon}}_{2-VM}(H;\mathcal{B})]
\label{vm-2}
\end{equation} 
\end{lemma} 

\begin{proof}
Suppose at time $t$ we {\bf do not} have parity on $\mathcal{B}$. Let $C_{t}$ be the resulting configuration.  Let $(a_{1},a_{2},\ldots, a_{n})$ be the vectors of label parities on $\mathcal{B}$ of all initial opinions. By our assumption there must be two different opinions $v_{1},v_{2}$ whose number of copies in $C_{t}|_{\mathcal B}$ have  different parities: $a_{i}\neq a_{j} \mbox{ (mod 2)}.$
We next apply the following trivial lemma

\begin{lemma} 
Conditional on being in state $C_{t}$, 
\[
Pr[\sum_{k\in {\bf D}} a_{k}\equiv 1 \mbox{ (mod 2)}]\geq 1/2-\epsilon.
\]
\end{lemma} 
\begin{proof} If $a_{i}$ is odd then including/excluding $i$ from $D$ changes the parity of $\sum_{j\in {\bf D}} a_{j}$.  
\end{proof} 

As a consequence we infer, similarly to the graph case \cite{aldous-fill-book}, that 
$Pr[C^{\bf D_{\epsilon}}_{2-\textit{VM}}(H;\mathcal{B})>t]\geq \big(\frac{1}{2}-\epsilon\big) Pr[C_{\textit{VM}}(H;\mathcal{B})> t].$
Finally, by Lemma~\ref{avglemma} $
E[C_{\textit{VM}}(H;\mathcal{B})]\leq \frac{2}{1-2\epsilon}\cdot E[C^{\bf D_{\epsilon}}_{2-VM}(H;\mathcal{B})]
$. 
\end{proof}

\subsection{Reachability and recurrence in the two-party voter model on odd-connected hypergraphs}
\label{sec:7:6}
The previous section motivates the study of recurrent states of the two-party voter model. For odd-connected hypergraphs the answer is especially simple:

\begin{theorem}
  Let $H$ be an odd-connected hypergraph and let $A$ be an initial state for the two-party voter model on $H$. Then the following hold: 
 \begin{itemize} 
 \item[(a).] If $H$ has an edge $e$ with $|e|$ odd then state {\bf 0}, in which all nodes have label zero, is reachable from $A$. 
 \item[(b).] If all edges $e$ of $H$ have even size then from any state different from {\bf 1} the system can reach state {\bf 0}. State {\bf 1} is another fixed point of the system. In this second state we have reached parity on all stabilizing sets $\mathcal{B}$ for the annihilating random walk on $H$. 
\end{itemize} 
\label{rec-voter}
\end{theorem}
\begin{proof}
Consider a state $B\neq {\bf 0}, {\bf 1}$ reachable from $A$. We will show that there is another state reachable from $A$ with strictly fewer ones than in $B$. 

Indeed, since $B$ is not an even dominating set, there must exist a hyperedge $e$ of $H$ such that $|e\cap B|$ is odd, in particular is nonzero. Consider the natural configuration $1_{B}$ with ones on $B$ and zero outside this set.

  By scheduling a node of $B\cap e$ we reach a configuration $1_{C}$ in the two-state voter model with a strictly lower number of ones. This happens, of course, since $|B\cap e|$ is odd.

  Since in the above argument $B$ was an arbitrary subset, it follows that state ${\bf 0}$ is reachable in the 2-state voter model from any initial configuration.
  
  The argument also works when $B= {\bf 1}$ if hypergraph $H$ has an edge $e$ with $|e|$ odd, since in this case by scheduling edge $e$ one can reach a state different from {\bf 1}. 
  
 In the second case, let $\mathcal{B}$ be a stabilizing state, let $B$ be its support, let $w_{1}$ be the vector which is 1 on $B$ and 1 outside of it. Adding all equations of satisfiable system $H(w_{1},0)$ we get 
 \[
 0 = \sum_{e\in E} |e|x_{e}=\sum_{v\in V} w_{1}[v]-0[v]= \sum_{v\in B} 1, 
 \]
 so $|B|$ must be even. Therefore in state {\bf 1} of the 2-state vector model we have reached parity on any state $\mathcal{B}$. 
\end{proof} 

As a sanity check, let us briefly discuss the implications of the previous result in the case when $H$ is a graph without loops. In this situation the voter model on $H$ reaches \textit{unanimity}, i.e. one of the opinions will dominate, and all other opinions will dissapear. In this case the two-party voter model will reach state {\bf 1} iff its initial state contains a 1 on the winning opinion of the voter model, and state {\bf 0} otherwise. 

The previous result has beneficial consequences for the analysis of the voter model (and, equivalently, of the lazy annihilating random walk) on odd-connected hypergraphs: instead of analyzing the  (difficult to control) condition of parity on some stabilizing state $\mathcal{B}$ one can instead bound  the weaker condition that the voter model reaches state {\bf 0} (or ones of states {\bf 0},{\bf 1} in case (b).) which then implies parity on \textbf{all} stabilizing sets $\mathcal{B}$.

\begin{definition} 
Let $H$ be an odd-connected hypergraph. Define 
\[
T_{parity}(H,D_{\epsilon}) =\min\{t\geq 1: D_{t}=\emptyset\},
\]
i.e. the minimum time when model $D_{\epsilon}$ is in configuration ${\bf 0}$ if $H$ is in case (a). of Theorem~\ref{rec-voter}, and  
\[
T_{parity}(H,D_{\epsilon}) =\min\{t\geq 1: D_{t}\in \{\emptyset,V\} \},
\]
in case (b). 
\end{definition} 

The previous discussion implies the following result: 
\begin{lemma}
\label{help}
We have 
\[
C^{\bf D_{\epsilon}}_{2-VM}(H;\mathcal{B})\leq T_{parity}(H,D_{\epsilon}).
\label{parity}
\]
\end{lemma}

\subsection{Completing the Proof of Theorem~\ref{vm}}
\label{sec:8}

To prove Theorem~\ref{vm} (1) and (2) we have to show that we have 
\begin{equation} 
max_{\mathcal{B}} E[C^{\bf D_{\epsilon}}_{2-VM}(H;\mathcal{B})]
 \leq T_{n}(H)
 \label{ub} 
\end{equation} 
for some function $T_{n}(H)$ with the properties from Theorem~\ref{vm}.

From Lemma~(\ref{parity}), all we have to prove is that in fact
$E[T_{parity}(H,D_{\epsilon})]\leq T_{n}(H)$ for a corrresponding function $T_{n}$. 
We first show that 

\begin{lemma} 
We have 
\[
P[\Delta N_{t}^{D_{\epsilon}}= -1]- P[\Delta N_{t}^{D_{\epsilon}}= 1]= \frac{|E^{-}(D_{t},\overline{D_{t}})|-|E^{+}(D_{t},\overline{D_{t}})|}{nk}
\]
\end{lemma} 
\begin{proof} 
In the two-party voter model. 
$N_{t}^{D_{\epsilon}}$ decreases by one exactly when the  chosen vertex $v$ has label 1 and the edge $e\ni v$ contains an odd number of nodes (including $v$ !) with label 1. Similarly, 
$N_{t}^{D_{\epsilon}}$ increases by one precisely when the chosen vertex $v$ has label 0 and the edge $e\ni v$ has an odd number of nodes with label 1. The number of distinct vertex-edge pairs in the two-party multiset voter model is precisely $kn$, since every vertex of $H$ has degree exactly $k$. The number of vertex-edge pairs that lead to an increase by 1 is  nothing but $|E^{+}(D_{t},\overline{D_{t}})|$, with $E^{+}$ having the meaning from  Definition~\ref{expansion-hyper}. 
\end{proof} 


We complete the proof of the upper bounds~(\ref{ub}) for Cases (i) and (ii) using drift arguments entirely similar to the ones used in the proof of Theorem~\ref{thm0} (i) and (ii). Instead of $\Delta u(t)$ we will control $\Delta N_{t}^{D_{\epsilon}}$, and we upper bound the expected time to hit {\bf 0}, conditional (when Case (b) of Theorem~\ref{rec-voter} applies) on not hitting state {\bf 1} (if this latter event happens parity is reached). 

 In Case (i) we have positive drift and the expected convergence time will be linear, whereas in Case (ii) we correct the expected running time of the unbiased lazy random walk with $\tau_{odd}(D(\Phi_{n}))$, due to inequality 
 \[
 P[\Delta N_{t}^{D_{\epsilon}}= -1]\geq \frac{1}{\tau_{odd}(D(\Phi_{n}))}
 \]

All other details are entirely similar. $\qed$

\section{Conclusions} 

Our main technical contribution has been to show that the use of hypergraph versions of particle systems renders the problem of analyzing the running time of XOR-SAT solvable via a reduction to drift analysis. 

Moreover, the generalizations of interacting particle systems we introduce have a number of interesting technical features: for instance, in contrast to the graph case (where it is nonincreasing), the number of live particles in annihilating random walks on hypergraphs may go up, and the structure of recurrent states is constrained by systems of linear equations similar to the ones used to analyze lights-out games \cite{scherphuis-lightsout}. On the other hand, in  coalescing random walks on hypergraphs there may be more than one copy of an initial "particle" and the process is naturally described using \textit{multisets}, rather than sets of particles. 

Our work could naturally be completed in many ways. First of all, we would like to remove some of the technical restrictions on the class of instances that can be analyzed with methods similar to ours. For instance, we would like to remove the condition of unique satisfiability from the statement of Theorem~\ref{thm0}. The condition of acyclicity of instances of XOR-SAT can probably be removed from the statement of Theorem~\ref{vm}. Concerning this result, it would be interesting to also obtain \textit{lower bounds} on the expected convergence time of $\mathit{WalkSAT}$ in terms of structural properties of the dual. Our method did not allow us to accomplish this goal, as it first replaced annihilating random walks by their lazy versions (that could be coupled to the lazy coalescing random walks, which are dual to the multiset voter model) and then replaced parity in the multiset voter model on the stabilizing set $B$ by a more tractable condition (that yielded an upper bound for its convergence time). 

On the other hand, our bounds are what Aldous \cite{aldous2013interacting} calls "crude upper bounds" based on "bottleneck statistics". More sophisticated methods and bounds on coalescence/annihilation exist \cite{cooper2013coalescing,cooper2016linear,kanade}. Perhaps similar bounds can be obtained in our framework, and are left as open problems. 

More importantly, it would be interesting to see if the running time of $\mathit{WalkSAT}$ and related local search procedures, can be analyzed on instances more interesting problems (e.g $k$-SAT) in terms of (suitably defined) "particle systems". 

Last but not least we believe that the explosive models defined in this paper deserve further independent study.  For instance, we would like to understand the structure of recurrent states in coalescing random walk model on hypergraphs when state ${\bf 0}$ is not reachable. Also, the long-term structure of the multiset voter model deserves, we believe, further clarifications, especially in the case (b). of Theorem~\ref{rec-voter} which is more similar to the graph case. 

\section{Acknowledgment} 

This research was supported in part by the project “Romanian Hub for Artificial Intelligence - HRIA”, Smart Growth, Digitization and Financial Instruments Program, 2021-2027, MySMIS no. 351416.

 \section*{Appendix} 
 
\subsection{Drift Theorems} 

To make the paper self-contained we first explicitly list the two drift theorems referenced in the paper:  

\begin{proposition}[Additive Drift Theorem]
\label{driftthm:additive}
Let $(X_t)_{t\geq 0}$ be a sequence of non-negative random variables with a finite state space $\mathcal{S} \subseteq \mathbb{R}_0^+$ such that $0 \in \mathcal S$. Let $T := \inf\{t \geq 0 \mid X_t =0\}$. 
\begin{enumerate}[(a)]
\item If there exists $\delta > 0$ such that for all $s\in \mathcal S \setminus\{0\}$ and for all $t\geq 0$,
\begin{align}\label{drifteq:additivecond}
\Delta_t(s):=E[X_t-X_{t+1}\mid X_t=s]\ge \delta,
\end{align}
then
\begin{align}\label{drifteq:additiveconc}
E[T] \leq \frac{E[X_0]}{\delta}.
\end{align}
\item If there exists $\delta > 0$ such that for all $s\in \mathcal S \setminus\{0\}$ and for all $t\geq 0$,
\begin{align}\label{drifteq:additivecondlower}
\Delta_t(s):=E[X_t-X_{t+1}\mid X_t=s]\le \delta,
\end{align}
then
\begin{align}\label{drifteq:additiveconclower}
E[T] \geq \frac{E[X_0]}{\delta}.
\end{align}
\end{enumerate}
\end{proposition} 

\begin{proposition}[Simplified Drift Theorem] \label{theo:modifieddrift}
  Let $X_t$, $t\ge 0$, be real-valued random variables describing a
	stochastic process over some state space. 
  Suppose there exist an interval $[a,b]\subseteq \mathbb{R}$, two
  constants $\delta,\epsilon>0$ and, possibly depending on
  $\ell:=b-a$, a function $r(\ell)$ satisfying  
  $1\le r(\ell) = o(\ell/\!\log(\ell))$ such that  
  for all $t\ge 0$ the following two conditions hold:
  \begin{enumerate}
  \addtolength{\itemsep}{0.5ex}
  \item $E(X_{t+1}-X_{t}\mid \mathcal{F}_t \;\wedge\; a< X_t <b)  \;\ge\; \epsilon$, 
  \item $\Pr(\lvert X_{t+1}-X_t\rvert\ge j \mid \mathcal{F}_t \;\wedge\; a< X_t) \;\le\;  \frac{r(\ell)}{(1+\delta)^{j}}$ for  $j\in \mathbb{N}_0$.
  \end{enumerate}
  Then there is a constant~$c^*>0$ such that for $T^*:={\min\{t\ge
  0\colon X_t\le a \mid \mathcal{F}_t \wedge X_0\ge b\}}$ it holds $\Pr(T^*\le
  2^{c^*\ell/r(\ell)}) = 2^{-\Omega(\ell/r(\ell))}$.
\end{proposition}

\subsection{Convergence of $\mathit{WalkSAT}$ on satisfiable instances} 

First, we state and prove the result alluded to in the introduction:

\begin{theorem}

Let $\Phi$ be a satisfiable instance of $k$-XOR-SAT specified by the system 
\begin{equation} 
A\cdot \vec{x} = \vec{b}. 
\label{sys:foo}
\end{equation} 
Let $X^{(1)}$ be an arbitrary vector. 
Denote by $Z$ a solution of system~(\ref{sys:foo}) at  minimal distance from $X^{(1)}$, and let $k=d_{H}(X^{(1)},Z)$. Then $Z$ is reachable from $X^{(1)}$ by legal moves of $\mathit{WalkSAT}$. 

\label{rec-xor}
\end{theorem} 
\begin{proof}
The proof is by induction on $k$: 

\begin{itemize}
\item[-] {\bf Case $k=0$.} Then $X^{(1)}=Z$ and there is nothing to prove. 
\item[-] {\bf Case $k=1$.} Then $X^{(1)}$ and $Z$ differ on a single variable $z$. Let $m$ be an equation containing $z$. Then $X^{(1)}$ does not satisfy $m$ (as $Z$, which only differs on $z$, does). Choosing equation $m$ and variable $z$ we reach $Z$ from $X^{(1)}$ by a legal move of $\mathit{WalkSAT}$. 
\item[-] {\bf Case $k\geq 2.$} If there is an equation $w$ not satisfied by $X^{(1)}$ (but satisfied by $Z$) then $w$ must contain a variable on which $X^{(1)}$ and $Z$ differ. Let $z$ be such a variable. Then by flipping the value of $z$ in $\mathit{WalkSAT}$ (by choosing clause $w$) 
one can reach from $X^{(1)}$ an assignment $X^{(2)}$ at Hamming distance $k-1$ from $Z$. $Z$ is still the solution of the system that is closest to $X^{(2)}$: if this were not the case (say $d(X^{(2)},Z^{\prime})<k-1$) 
then we would have $d(X^{(1)},Z^{\prime})<k$, contradicting the hypothesis about the minimality of $Z$. 
  
By the induction hypothesis one can reach $Z$ from $X^{(2)}$ by legal moves of $\mathit{WalkSAT}$. Therefore the same statement is true from $X^{(1)}$. 
\end{itemize}
\end{proof}

\subsection*{Proof of Lemma~\ref{mc-hit}}

Consider the Lyapunov function
$
V(i)=i(2n+1-i), 0\leq i\leq n.$
Clearly $V(0)=0$ and $V(i)\geq 0$ for all $i$. For $1\leq i<n$, using the transition probabilities and the symmetry of the left and right moves,
\[
\begin{aligned}
\mathbb E[V(X_{t+1})-V(X_t)\mid X_t=i]
&=
p_i\bigl(V(i-1)-V(i)\bigr)
+
p_i\bigl(V(i+1)-V(i)\bigr) \\
&=
p_i\bigl(V(i-1)+V(i+1)-2V(i)\bigr).
\end{aligned}
\]
Since $V$ is a quadratic function with second difference $V(i-1)+V(i+1)-2V(i)=-2$, 
we obtain $\mathbb E[V(X_{t+1})-V(X_t)\mid X_t=i]
=
-2p_i
\leq -2b.$

At the reflecting endpoint $n$, we have
$
V(n)-V(n-1)=2.
$
Therefore, if the chain moves from $n$ to $n-1$ with probability at least $b$, while otherwise staying at $n$, then
$
\mathbb E[V(X_{t+1})-V(X_t)\mid X_t=n]
\leq
-2b.
$ Hence, for all nonzero states,
$
\mathbb E[V(X_{t+1})-V(X_t)\mid X_t]\leq -2b.
$
Equivalently, $
M_t := V(X_{t\wedge \tau_0})+2b(t\wedge \tau_0)$ 
is a supermartingale. By optional stopping,
$\mathbb E_x[M_t]\leq M_0=V(x).
$
Thus
$
2b\,\mathbb E_x[t\wedge \tau_0]
\leq V(x).$ (here $t\wedge \tau_0$ means truncating  hitting time $\tau_0$ by time $t$).

Letting $t\to\infty$ and using monotone convergence gives $
2b\,\mathbb E_x[\tau_0]\leq V(x).$
Therefore $\mathbb E_x[\tau_0]\leq \frac{V(x)}{2b}
=
\frac{x(2n+1-x)}{2b}.$
For $x\in 0,\ldots n$, this is $O(n^2/b)$.

\subsection*{Counterexamples for reachability} 

We give two types of counterexamples to the converse of Lemma~\ref{if}. The first one is the setting for which a partial converse was (incorrectly) claimed in \cite{istrate-balance}: 
connected hypergraphs without graph edges. Interestingly enough, for a modulo-$p$ generalization of our dynamics (with $p\geq 3$) such a counterexample does not exist \cite{istrate2015reachability}. 

The second counterexample shows that the failure of the converse implication is not specific to hypergraphs: even on graphs the sufficient condition fails to be necessary.  

In the sequel, we will use notation $\Vdash$ to denote the transitive closure of relation $\vdash$: Specifically, we write $C_{1}\Vdash C_{2}$ iff $C_{2}$ is reachable from $C_{1}$  
by a sequence of moves (and $C_1\not\Vdash C_2$ in the opposite case).  

\begin{figure}
\begin{center} 
\begin{minipage}{.37\textwidth}
\begin{center}
\includegraphics[width=5cm]{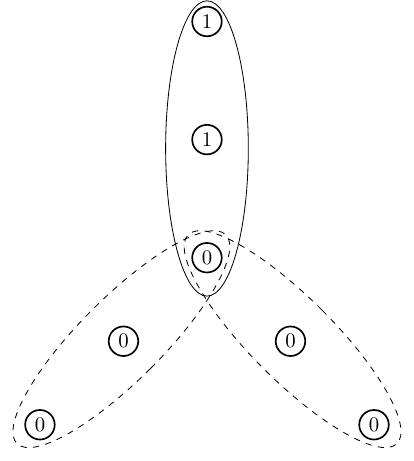}
\end{center}
\end{minipage}
\begin{minipage}{.1\textwidth}
\Huge{
\[
\nVdash
\]}
\end{minipage} 
\begin{minipage}{.37\textwidth}
\begin{center}
\includegraphics[width=5cm]{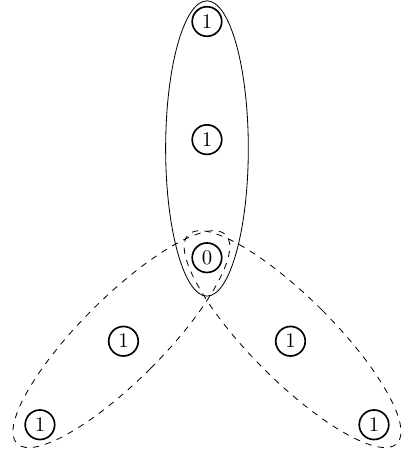}
\end{center}
\end{minipage}
\end{center} 
\caption{Unreachability in a hypergraph with no graph edges.}
\label{two-triangles}
\end{figure}

\begin{figure}
\begin{center} 
\begin{minipage}{.37\textwidth}
\begin{center}
\includegraphics[width=3cm]{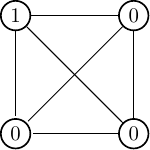}
\end{center}
\end{minipage}
\begin{minipage}{.1\textwidth}
\Huge{
\[
\nVdash
\]}
\end{minipage} 
\begin{minipage}{.37\textwidth}
\begin{center}
\includegraphics[width=3cm]{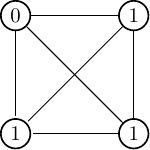}
\end{center}
\end{minipage}
\end{center} 
\caption{Unreachability in  a graph.}
\label{two-graph} 
\end{figure}

\begin{theorem} 
There exist
\begin{itemize} 
\item[(a). ]
a connected hypergraph $H$ that contains no graph edges,
and
\item[(b). ] a connected graph (i.e. all hyperedges have size two) $H$,   
\end{itemize}
as well as two configurations $w_{1},w_{2}$ on $H$ such that system $H(w_{1},w_{2})$ has solutions in ${\bf Z}_{2}$, yet $w_{1}\not\vdash w_2$ (i.e. $w_2$ is not reachable in $H$ from $w_{1}$).  

\label{graph} 
\end{theorem} 

\begin{proof} 
\begin{enumerate} 

\item

Let $H$ be a hypergraph consisting of three hyperedges $e_{1},e_{2},e_{3}$ sharing a common vertex (Figure~\ref{two-triangles}). Let $w_{1}, w_{2}$ be the configuration described in that figure: the private vertices of $e_{2}$ (displayed with a solid line in Figure~\ref{two-triangles}) have initial value 1 in $w_{1}$, 
all other vertices being 0. On the other hand $w_{2}$ takes value 0 on the shared vertex and 1 everywhere else. 

It is easy to see that system $H(w_{1},w_{2})$ has a solution $z$ with $z(e_{1})=z(e_{3})=1$ and $z(e_{2})=0$. Yet $w_{2}$ is not reachable from $w_{1}$. Indeed hyperedges with three labels of one have no preimage. So the only preimages of state $w_{2}$ are itself and the three ones obtained by flipping labels on one hyperedge. 
\item

Let $H$ be the complete graph $K_{4}$ and let $w_{1}$ be 1 at a single vertex $v$ (Figure~\ref{two-graph}). Let $w_{2}$ be the configuration with ones at \textit{every vertex but $v$}. System $H(w_{1},w_{2})$ has solution $z_{e}=1$ for every edge $e$, yet 
$w_{2}$ is {\bf not} reachable from $w_{1}$, as $w_{1}$ has a single one and $w_{2}$ has three, but on a graph the number of ones does not increase. 
\end{enumerate} 
\end{proof}

\end{document}